\documentclass[12pt]{article}
\RequirePackage[authoryear]{natbib}

\RequirePackage{amsthm,amsmath,amsfonts,amssymb}
\RequirePackage[colorlinks,citecolor=blue,urlcolor=blue]{hyperref}
\RequirePackage{graphicx}

\usepackage{float}
\usepackage{enumitem}
\usepackage{url}
\usepackage{bbm}
\usepackage[nottoc,numbib]{tocbibind}
\usepackage{booktabs,caption,subcaption}
\usepackage[flushleft]{threeparttable}
\usepackage{multirow}

\usepackage[pdftex,dvipsnames]{xcolor}
\usepackage{comment}

\usepackage[margin=1in]{geometry}

\usepackage{macros}

\newtheorem{theorem}{Theorem}

\newtheorem{lemma}{Lemma}

\theoremstyle{definition}

\begin{document}

%\jmlrheading{23}{2022}{1-\pageref{LastPage}}{1/21; Revised 5/22}{9/22}{21-0000}{}

\title{\bf Transfer Learning with Uncertainty Quantification: Random Effect Calibration of Source to Target (RECaST)}
 \author{Jimmy Hickey, Jonathan P. Williams, Emily C. Hector\hspace{.2cm}\\
  Department of Statistics, North Carolina State University}
  \date{}
 \maketitle
 
 \bigskip
\begin{abstract}
Transfer learning uses a data model, trained to make predictions or inferences on data from one population, to make reliable predictions or inferences on data from another population. Most existing transfer learning approaches are based on fine-tuning pre-trained neural network models, and fail to provide crucial uncertainty quantification. We develop a statistical framework for model predictions based on transfer learning, called {\em RECaST}. The primary mechanism is a Cauchy random effect that recalibrates a source model to a target population; we mathematically and empirically demonstrate the validity of our RECaST approach for transfer learning between linear models, in the sense that prediction sets will achieve their nominal stated coverage, and we numerically illustrate the method's robustness to asymptotic approximations for nonlinear models. Whereas many existing techniques are built on particular source models, RECaST is agnostic to the choice of source model.  For example, our RECaST transfer learning approach can be applied to a continuous or discrete data model with linear or logistic regression, deep neural network architectures, etc.  Furthermore, RECaST provides uncertainty quantification for predictions, which is mostly absent in the literature.  We examine our method's performance in a simulation study and in an application to real hospital data.
\end{abstract}

\noindent%
{\it Keywords:} Domain adaptation, Electronic health records, Informative Bayesian prior, Model calibration.%, Random effect model.
\vfill

\section{Introduction}
\label{Section: Introduction}

The use of artificial intelligence and machine learning (ML) is frequently limited in practice by a shortage of available training data and insufficient computational resources. To address these difficulties, transfer learning has developed as a powerful idea for leveraging the resources at leading institutions such as research hospitals (e.g., institutions having high quality data, exceptional research clinicians, high performance computing environments, etc.) to facilitate implementation of ML technologies in resource scarce settings such as small or rural hospitals.  Developments in transfer learning methodologies are necessary to overcome resource allocation inequities, and they will likely drive the next decade of innovation in ML technologies.  

Transfer learning consists broadly of two elements.  The first is one or more {\em target} population(s) of interest that are associated with data sets for which there are resource limitations preventing the training of sophisticated models (e.g., a small hospital).  The second is a {\em source} population (or populations) that is separate but in some way related to the target population.  The source is associated with extensive data and/or resources for training sophisticated ML models. The premise of transfer learning is to use trained source models to aid in the training of target models.  The source and targets are each composed of two components: a {\em domain}, denoted $\mathcal{D}$, and a {\em task}, denoted $\mathcal{T}$. A domain $\mathcal{D} := \{\mathcal{X}, P(x)\}$ consists of a feature space $\mathcal{X}$ and a marginal probability distribution $P(x)$ over $x \in \mathcal{X}$. A task $\mathcal{T} := \{\mathcal{Y}, P(y\mid x)\}$ is composed of a label space $\mathcal{Y}$ and a conditional distribution $P(y\mid x)$ over $y \in \mathcal{Y}$ given $x \in \mathcal{X}$.  Traditional ML is described by the source and target sharing the same domain, $\mathcal{D}_{S} = \mathcal{D}_{T}$, and sharing the same task, $\mathcal{T}_{S} = \mathcal{T}_{T}$.  Transfer learning problems arise when the source and target domains and/or the source and target tasks are similar but different.

We propose a new Bayesian transfer learning framework termed {\em RECaST} that focuses on transfer learning problems characterized by shared domains and label spaces, but with potentially different feature-to-label mappings $P(y\mid x)$. For example, source and target hospitals might record largely the same patient data features, but nuances in clinician practices/procedures, inconsistencies in data quality, population disparities, etc. may affect the suitability of using the source mapping as the target mapping.
Two primary advantages of RECaST are its scalability, requiring estimation of only 2-3 parameters with no tuning parameters, and that it is agnostic to the source model specification. Importantly, RECaST only requires the source model and parameter estimates, not the source data itself; this is an immense benefit to applications with privacy concerns, such as with medical data. In addition to common prediction summaries, RECaST also provides uncertainty quantification in the form of posterior predictive credible sets, which is largely absent in current methods. Further, RECaST is asymptotically valid in the canonical case of distinct source and target Gaussian linear models, in that the coverage of prediction sets are guaranteed to asymptotically achieve their stated nominal level of significance.

To evaluate our proposed RECaST approach, we design synthetic simulation studies with both continuous and binary response data reflecting a variety of difficulty levels of transfer learning problems.  Next, we investigate the performance of RECaST in real data simulations that arise by permuting real patient data from the multi-center eICU Collaborative Research Database \citep{pollard2018eicu}.  A variety of both point-valued and set-valued prediction metrics are considered, including the empirical coverage of prediction sets.  The performance of RECaST is compared to other state-of-the-art transfer learning approaches, including deep learning approaches and approaches designed specifically for clinical data \citep[e.g.,][]{wiens2014}.  The approach of \cite{wiens2014} was developed to improve predictions for the risk of hospital associated infection with Clostridium difficile in target hospitals using a source data augmented transfer learning strategy.

General survey papers on transfer learning topics include \cite{pan2009, lu2015, weiss2016, dube2020}.
For hospital disease risk and mortality prediction problems, \cite{wiens2014}, \cite{gong2015}, and \cite{desautels2017} propose transfer learning approaches based on training algorithms using a learned weighted combination of source and target patient observations.  These methods learn many parameters and require access to the source data -- limitations that are avoided by RECaST.  In \cite{paul2016}, \cite{raghu2019}, and \cite{AHISHAKIYE2021118}, approaches are considered to improve classification accuracy for medical imaging tasks using pre-trained deep neural networks (DNNs) on the ImageNet database \citep{deng2009}.  In the context of ICU patient monitoring, in \cite{shickel2021} a data augmenting-based transfer learning approach is built for fitting a single-layer recurrent neural network trained on electronic health records (EHR) and wearable device data.  Their model is limited in scope to only predicting the binary response of successful versus unsuccessful discharge from a hospital.  Implemented in \cite{gao_cui2021} is a transfer learning strategy for precision medicine in survival analysis with clinical omics data sets via freezing layers of a pre-trained Cox neural network.  Developed in \cite{lee2012} is a method using support vector machines to predict surgical mortality.  Another approach, from \cite{gu2022syntl}, is to generate additional synthetic target data from a source data set and adjust for heterogeneity in order to predict extreme obesity from medical records and genomics data.  An example of low-dimensional representation transfer learning is given in \cite{maurer2016}, and {\em online} transfer learning is considered in \cite{zhao2014,wu2017}.  These applied methods are useful in modeling specific pieces of EHR data for prediction, but lack uncertainty quantification.  Additionally, some require the learning of many parameters and access to the entire source data set.

Bayesian transfer learning adaptations include \cite{baxter1998}, \cite{raina2006}, \cite{wohlert2018}, \cite{bueno2019}, \cite{chandra2020}, \cite{yang2020}, \cite{zhou2020}, \cite{abba2022}; all except \cite{baxter1998} and \cite{raina2006} are based on priors specified from neural network models fitted to source data sets.  A posterior distribution fitted to a source DNN model is used as a prior on the parameters for the target task in \cite{wohlert2018}, and the model is trained using mean field variational Bayes \citep[for a reference on variational Bayes, see][]{zhang2018}.  Boosting approaches to transfer learning are considered by \cite{freund1999}, \cite{dai2007}, and \cite{desautels2017}.  In \cite{abba2022}, a penalized complexity prior between the source and target tasks is considered.  These surveyed strategies are predominantly non-model based, purely empirical, and lack a unified underlying framework.  Moreover, those that focus on fine-tuning pre-trained neural network models on a target data set require the source model to be a neural network, and often fail to provide crucial uncertainty quantification.

Recently, there have been efforts to investigate theoretical properties of transfer learning approaches.  For instance, a learning method based on LASSO for high-dimensional penalized linear regression is considered in \cite{li2020transfer}, while diminishing the effect of {\em negative transfer}. Negative transfer occurs when including source data negatively impacts the performance on target data.  In a similar setting, asymptotically valid confidence intervals for generalized linear model parameters in high-dimensional transfer learning problems are established in \cite{tian2022transfer}.  This technique is adapted to a more complicated federated transfer learning setting in \cite{li2021targeting}.  A parameter is defined in \cite{cai2021transfer} to calculate an ``effective sample size"  to quantify total amount of information that can be transferred when the source and target conditional distributions differ. This approach is extended in \cite{reeve2021adaptive}, where assumptions are relaxed on the relationship between the source and target conditional distributions. \cite{Hector-Martin} propose and study the inferential properties of an information-driven shrinkage estimator that is robust to heterogeneity between source and target feature-to-label mappings but assumes this mapping is of the same parametric form. These methods offer more mathematically rigorous motivations, but are restrictive in their modeling options. Such restrictions are eliminated in our proposed framework.
 
The remainder of our paper is organized as follows.  In Section \ref{section: recast framework}, we develop the theoretical basis for RECaST and its uncertainty quantification.  We then develop Bayesian parameter estimation and prediction procedures in both the continuous and binary response cases in Sections \ref{section: continuous} and \ref{Section: Binary Response Data}, respectively.  We conduct extensive simulation studies in Section \ref{section: Simulation Study} by exploring transfer learning problems of a range of difficulties.  Section \ref{Section: eICU} considers a real data analysis for predicting shock in ICU data.  Section \ref{Section: conclusion} concludes.  Proofs and computational details are provided in the Appendix.

\section{RECaST Framework}\label{section: recast framework}

Our transfer learning problem is defined by the following four assumptions: (i) there is a well-developed structural component of the prediction model for the source domain denoted by $f(\btheta_{S}, x_{S})$; (ii) there exist ample source data for estimating the parameter(s) $\btheta_{S}$; (iii) $\mathcal{X}_{S} = \mathcal{X}_{T}$, and the structural component of the target prediction model, denoted by $g(\btheta_{T}, x_{T})$, is believed to be {\em similar} to $f(\btheta_{S}, x_{T})$; and (iv) there does not exist sufficient target data for reliably estimating the parameter(s) $\btheta_{T}$.  The notion of {\em similarity} will be defined in the construction of our RECaST framework for transfer learning, presented next.

Denote the forward data-generating representations of $P(y_{S}\mid \bx_{S})$ and $P(y_{T}\mid \bx_{T})$, respectively, by 
\begin{equation}\label{equation: dge}
\begin{split}
Y_{S} &= h\big\{f(\btheta_{S}, \bx_{S}), U_{S}\big\} \ \ \text{and} \\
Y_{T} & = h\big\{g(\btheta_{T},\bx_{T}), U_{T}\big\}, 
\end{split}
\end{equation}
where $\mathcal{X}_{S} = \mathcal{X}_{T} = \R^{p}$, and $U_{T}$ and $U_{S}$ are independent auxiliary random variables.  For example, if 
\[
\begin{split}
f(\btheta_{S}, \bx_{S}) & = \bx_{S}^{\top}\btheta_{S}, \\
h(\bx_{S}^{\top}\btheta_{S}, U_{S}) & = \bx_{S}^{\top}\btheta_{S} + U_{S}, \ \ \text{and} \\
U_{S} & \sim \mathcal{N}(0,1),                   
\end{split}
\]
then $Y_{S} \sim \mathcal{N}(\bx_{S}^{\top}\btheta_{S}, 1)$.  Or in the case of binary outcome data, for example, if 
\[
\begin{split}
f(\btheta_{S}, \bx_{S}) & = \text{expit}(\bx_{S}^{\top}\btheta_{S}), \\
h(\bx_{S}^{\top}\btheta_{S}, U_{S}) & = \mathbf{1}\{U_{S} < \text{expit}(\bx_{S}^{\top}\btheta_{S})\}, \ \ \text{and} \\
U_{S} & \sim \text{Uniform}(0,1), \\
\end{split}
\]
then $Y_{S} \sim \text{Bernoulli} \{ \text{expit}(\bx_{S}^{\top}\btheta_{S}) \}$, where $\text{expit}(z) := e^{z}/(1+e^{z})$.  The {\em similarity} between the source and target that makes this a formulation of a transfer learning problem is determined by how well the structural component $f(\btheta_{S}, \bx_{T})$ of the source model approximates the structural component $g(\btheta_{T}, \bx_{T})$ of the target model.  

Accordingly, transfer learning should be effective if $\beta :=  g(\btheta_{T},\bx_{T}) / f(\btheta_{S},\bx_{T}) \approx 1$, and sufficient source data is available for reliable estimation of $\btheta_{S}$; in fact, the source and target models are identical if $\beta = 1$.  Assuming $f(\btheta_{S},\bx_{T}) \ne 0$ {\em almost surely} (a.s.), it follows a.s. that 
\begin{equation}\label{transfer_learning_model}
Y_{T,i} = h\big\{\beta_{i} \cdot f(\btheta_{S},\bx_{T,i}), U_{T,i}\big\},
\end{equation}
for $i \in \{1, \dots, n_{T}\}$, where $Y_{T,1}, \dots, Y_{T,n_{T}}$ is an independent sample of $n_{T}$ target labels with associated features $\bx_{T,1}, \dots, \bx_{T,n_{T}}$, and $\beta_{i} := g(\btheta_{T},\bx_{T,i}) / f(\btheta_{S},\bx_{T,i})$.  The identity given by Equation (\ref{transfer_learning_model}) is further motivated by the fact that, for linear source and target models, if we assume $\bx_{T,1}, \dots, \bx_{T,n_{T}} \overset{\text{iid}}{\sim} \mathcal{N}_{p}({\bf 0}, \bI_{p})$, then by Lemma \ref{lemma: cauchy} (a well-known result for which we provide a proof in Appendix \ref{appendix: proof}, for convenience), $\beta_{i} = (\bx_{T,i}^{\top}\btheta_{T}) / (\bx_{T,i}^{\top}\btheta_{S}) \sim \text{Cauchy}(\delta, \gamma)$, with
\begin{align*}
\delta & =  \frac{ \btheta_{T}^{\top} \btheta_{S} }{ \norm{\btheta_{S}}^{2} }, \ \ \text{and} \\
\gamma & = \frac{ 1 }{ \norm{\btheta_{S}}^{2} }  \sqrt{ \norm{\btheta_{S}}^{2} \norm{\btheta_{T}}^{2} - (\btheta_{T}^{\top} \btheta_{S})^{2} }.
\end{align*}

\begin{lemma}\label{lemma: cauchy}
For any $\ba, \bb \in \mathbb{R}^{p}$, if $\bx \sim \mathcal{N}_{p}({\bf 0}, \bI_{p})$ then $(\bx^{\top}\ba) / (\bx^{\top}\bb) \sim \text{Cauchy}(\delta, \gamma)$, with $\delta = \ba^{\top} \bb / \norm{\bb}^{2}$ and $\gamma = \norm{\bb}^{-2} \sqrt{ \norm{\bb}^{2} \norm{\ba}^{2} - (\ba^{\top} \bb)^{2} }$.
\end{lemma}

In practice, we assume without loss of generality that features have been centered and scaled to have mean zero and unit variance.  Central limit theory supports the Gaussian approximation for more complex, nonlinear models (i.e., for the large $p$ scenarios that characterize modern ML approaches).  Specifically, appealing to the Lyapunov or Lindeberg central limit theorem gives Gaussian approximations for the distributions of $\bx_{T,i}^{\top}\btheta_{S} / \|\btheta_{S}\|^{2}$ and $\bx_{T,i}^{\top}\btheta_{T} / \|\btheta_{T}\|^{2}$.  For more general assumptions on $f$ and $g$, first order Taylor approximations motivate $f(\btheta_{S}, \bx_{T,i}) \approx \bx_{T,i}^{\top}\btheta_{S}$ and $g(\btheta_{T}, \bx_{T,i}) \approx \bx_{T,i}^{\top}\btheta_{T}$.  The edge case with $\gamma \to \infty$ describes a situation in which there is no link between the source and target domains.  Assuming $\gamma < \infty$, the RECaST model specified by Equation \eqref{transfer_learning_model} with random effect $\beta_{i} \sim \text{Cauchy}(\delta, \gamma)$ fully characterizes the {\em similarity} between the source and target domains.  In addition to being the exact distribution in the linear model case with Gaussian features, the Cauchy distribution also provides benefit through its heavy tails.  This attribute allows $\beta_i$ to capture large disparities between source and target data sets, improving the frequentist coverage of resulting prediction sets.

Estimating parameters of Cauchy distributions is a notoriously difficult problem since the heavy tails allow outlying events to happen with relatively high probability \citep{schuster2012parameter}.  Some estimation procedures focus on estimating solely the location parameter \citep{zhang2010highly} or the scale parameter \citep{kravchuk2012hodges}, but rarely both.  \cite{fegyverneki2013simple} explores the trade-off between using simple robust estimators, for both parameters, which are less asymptotically efficient than the maximum likelihood estimators.  Recently, limit theorems are established in \cite{akaoka2022limit} for quasi-arithmetic means for point estimation in cases where the strong law of large numbers fails, such as with Cauchy random variables.  The fact that the Cauchy distribution appears in our work speaks to the difficulty of a transfer learning problem.

There are three primary advantages of our RECaST transfer learning model formulation in Equation \eqref{transfer_learning_model} with random effect $\beta_{i} \sim \text{Cauchy}(\delta, \gamma)$.  First, regardless of the complexity of the source model (e.g., $f(\btheta_{S},\cdot)$ could represent a DNN with millions of parameters in $\btheta_{S}$ trained on extensive source data), RECaST only ever requires estimation of the parameters $\delta$ and $\gamma$, and perhaps a scale parameter associated with $U_{T,i}$ through $h(\cdot,U_{T,i})$. Existing transfer learning methods require either estimation of $\btheta_{T}$ (often via fine-tuning from an estimate of $\btheta_{S}$) or learning of $n_{T} + n_{S}$ weights for pooling the source and target data, where $n_{S}$ is the number of source training labels.  The scalability of our approach cannot be overstated.  Second, RECaST needs no source data, only requiring the estimated source parameters $\widehat{ \btheta }_{S}$.  Such a feature is vital in applications such as with medical data where privacy constraints place legal and ethical barriers to accessing certain data sets.  Third, RECaST naturally facilitates uncertainty quantification of target label predictions via the construction of prediction sets.  We assess the empirical coverage with respect to the nominal significance level of the prediction sets.  The following two sections propose a Bayesian framework for estimation of the posterior predictive distribution of target labels in the continuous and binary response settings, respectively.

\section{Continuous Response Data}\label{section: continuous}

\subsection{Model and Estimation}
\label{subsection: continuous model and estimation}

Assume that $Y_{S,1}, \dots, Y_{S,n_{S}}$ and $Y_{T,1}, \dots, Y_{T,n_{T}}$ are mutually independent, continuous random variables generated according to source and target models, respectively, as expressed in Equation (\ref{equation: dge}). Also assume that an estimator $f(\widehat{\btheta}_{S},\bx)$ is available for any feature vector $\bx \in \X_{S} = \X_{T}$, where $\widehat{\btheta}_{S}$ is an estimator of $\btheta_{S}$ based on $Y_{S,1}, \dots, Y_{S,n_{S}}$.  In the continuous response setting, a natural choice for the $h$ function in the RECaST model, defined by Equation \eqref{transfer_learning_model}, is the Gaussian innovation formulation,
\[
Y_{T,i} = \beta_{i}\cdot f(\widehat{\btheta}_{S},\bx_{T,i}) + \sigma\cdot U_{T,i},
\]
independently for $i \in \{1,\dots,n_{T}\}$, where $U_{T,i} \sim \mathcal{N}(0, 1)$, $\sigma > 0$ is a scaling parameter to be learned from the target data, and $\beta_{i} \sim \text{Cauchy}(\delta, \gamma)$.  

Denoting by $\pi(\delta, \gamma, \sigma)$ a prior density on $(\delta, \gamma, \sigma)$, a posterior distribution of the parameters $(\delta, \gamma, \sigma)$ can be expressed as
\begin{align*}
& \pi\big(\delta, \gamma , \sigma \mid y_{T,1}, \dots, y_{T,n_{T}}, \widehat{\btheta}_{S}\big) \\
& = \int_{\R}\dots \int_{\R} \pi\big(\delta, \gamma, \sigma, \beta_{1}, \dots, \beta_{n_{T}} \mid y_{T,1}, \dots, y_{T,n_{T}}, \widehat{\btheta}_{S}\big) \ d\beta_{1} \dots d\beta_{n_{T}} \\
& \propto \pi(\delta, \gamma, \sigma) \cdot \int_{\R}\dots \int_{\R} \prod_{i=1}^{n_{T}}\Big[\mathcal{N}\big\{ y_{T,i} \mid \beta_{i} f(\widehat{\btheta}_S, \bx_{T,i}), \sigma^{2} \big\} \cdot \text{Cauchy}( \beta_{i} \mid \delta, \gamma )\Big] \ d\beta_{1} \dots d\beta_{n_{T}} \\
& = \pi(\delta, \gamma, \sigma) \cdot \prod_{i=1}^{n_{T}} \int_{\R} \mathcal{N}\big\{ y_{T,i} \mid \beta_{i} f(\widehat{\btheta}_S, \bx_{T,i}), \sigma^{2} \big\} \cdot \text{Cauchy}( \beta_{i} \mid \delta, \gamma ) \ d\beta_{i} \\
& = \pi(\delta, \gamma, \sigma) \cdot \prod_{i=1}^{n_{T}} \int_{\R} \frac{ \text{Cauchy}( \beta_{i} \mid \delta, \gamma ) }{ \mid f(\widehat{\btheta}_S, \bx_{T,i}) \mid } \cdot \mathcal{N}\bigg\{ \beta_{i} \ \mid \ \frac{  y_{T,i} }{ f(\widehat{\btheta}_{S}, \bx_{T,i}) }, \frac{ \sigma^{2} }{ f^{2}(\widehat{\btheta}_S, \bx_{T,i})  } \bigg\} \ d\beta_{i},
\numberthis \label{equation: continuous posterior}
\end{align*}
where the univariate integrals in the last expression can be evaluated numerically.  Next, the posterior predictive distribution of the label $\widetilde{Y}_{T}$ associated with some new target feature vector $\widetilde{\bx}_{T}$ can be derived as the marginal distribution of 
\begin{align*}
& \pi\big(\widetilde{y}_{T}, \widetilde{\beta}, \sigma, \delta, \gamma \mid y_{T,1}, \dots, y_{T,n_{T}}, \widehat{\btheta}_{S} \big) \\ 
%& = \mathcal{N}\{\widetilde{y}_{T} \mid \widetilde{\beta}f( \widehat{\btheta}_{S}, \widetilde{\bx}_{T}), \sigma^{2}\} \cdot \pi\big(\widetilde{\beta}, \sigma, \delta, \gamma \mid y_{T,1}, \dots, y_{T,n_{T}}, \widehat{\btheta}_{S} \big) \\
&  = \mathcal{N}\{\widetilde{y}_{T} \mid \widetilde{\beta}f( \widehat{\btheta}_{S}, \widetilde{\bx}_{T}), \sigma^{2}\}\cdot \text{Cauchy}(\widetilde{ \beta } \mid \delta, \gamma) \cdot \pi\big(\delta, \gamma, \sigma \mid y_{T,1}, \dots, y_{T,n_{T}}, \widehat{\btheta}_{S}\big). \numberthis \label{equation: continuous posterior predictive}
\end{align*}
We specify a canonical prior on $(\delta, \gamma, \sigma)$ as 
\[
\pi(\delta, \gamma , \sigma) = \mathcal{N}(\delta \mid 1, \sigma_{\delta}^{2}) \cdot \log \mathcal{N}(\gamma \mid a,b) \cdot \log \mathcal{N}(\sigma \mid c,d),
\]
with diffuse choices of the hyperparameters $\sigma_{\delta}, a, b, c, d$.

\subsection{Remarks on Implementation}\label{subsection: continuous remarks on implementation}

To estimate the posterior distribution given in Equation \eqref{equation: continuous posterior}, we implement a random walk Metropolis-Hastings algorithm, numerically solving the univariate integrals with the \texttt{Julia} package \texttt{QuadGK} \citep{johnson2020quadgk}.  Furthermore, by expressing these integrals as expectations with respect to a Gaussian distribution (i.e., the final expression in Equation \eqref{equation: continuous posterior}), we show that they are numerically equivalent to definite integrals from $-39$ to $39$.  See Appendix \ref{appendix: Bounding Continuous Integral} for the mathematical details of this bound.  This substantially reduces the computational overhead for the numerical integration, and the fact that $n_{T}$ is assumed to be small for transfer learning problems avoids concerns about scalability.  Next, in Algorithm \ref{algorithm: continuous RECaST prediction}, we propose a procedure for drawing samples from the posterior predictive distribution described by Equation (\ref{equation: continuous posterior predictive}).  We showcase the effectiveness of these prosed computational strategies in a variety of simulation scenarios in Section \ref{sims_continuous_response}.  Posterior predictive credible sets can be constructed as usual in Bayesian inference, from the highest posterior density regions calculated via the empirical quantiles of the sampled posterior predictive values.
\begin{algorithm}[H]
 \caption{\footnotesize RECaST posterior predictive sampling: continuous response data}\label{algorithm: continuous RECaST prediction}
 \textbf{Input:} $\widetilde{\bx}_{T}$, samples from $\pi\big(\delta, \gamma, \sigma \mid y_{T,1}, \dots, y_{T,n_{T}}, \widehat{\btheta}_{S}\big)$, and sample sizes $n_{\text{post}}$, $n_{\beta}$, and $n_{Y}$ \\
 \textbf{Output:} A sample of values from $\pi\big(\widetilde{y}_{T} \mid y_{T,1}, \dots, y_{T,n_{T}}, \widehat{\btheta}_{S} \big)$
 \begin{algorithmic}

 \For{$i \gets 1$ to $n_{\text{post}}$}
     \State $\delta, \gamma, \sigma \gets \text{random} \big\{ \pi\big( \delta, \gamma, \sigma \mid y_{T,1}, \dots, y_{T,n_{T}}, \widehat{\btheta}_{S}\big)  \big\}$
     \For{$j \gets 1$ to $n_{\beta}$}
         \State $\widetilde{ \beta } \gets \text{random}\big\{ \text{Cauchy}(\delta, \gamma)  \big\}$
         \For{$k \gets 1$ to $n_{Y}$}
             \State $\widetilde{Y}_{T} \gets \text{random}\big[ \mathcal{N}\big\{\widetilde{ \beta } f(\widehat{ \btheta }_{S}, \widetilde{\bx}_{T}), \sigma^{2} \big\}\big]$
         \EndFor
     \EndFor
 \EndFor
 \end{algorithmic}
\end{algorithm}

\subsection{Theoretical Guarantees}
\label{subsection: Continuous Theoretical Guarantees}

In this section, we establish the asymptotic validity of our proposed posterior predictive credible sets in the case of linear source and target models with independent Gaussian innovations.  Here, asymptotic validity means that the empirical coverage of a $1 - \alpha$ level prediction credible set attains $1 - \alpha$ level coverage, asymptotically in $n_{T}$, as described by the result of Theorem \ref{theorem: convergence}.  Our mathematical proof of this result and of all supporting results are organized in Appendix \ref{appendix: proof}.

Suppose that $Y_{S,j} \sim \mathcal{N}(\bx_{S,j}^{\top} \btheta_{S}, \sigma^{2})$, independently for $j \in \{1,\dots,n_{S}\}$.  In the class of transfer learning problems we consider, it is assumed that consistent or meaningful estimators are available for all source model parameters, and that ample data/resources are available for estimating them.  Accordingly, assume that $n_{S}$ is sufficiently large so that $\btheta_{S}$ and $\sigma$ are regarded as known.  Next, assume that $Y_{T,1}, \dots, Y_{T,n_{T}} \overset{\text{iid}}{\sim} \mathcal{N}(\widetilde{\bx}^{\top} \btheta_{T}, \sigma^{2})$, for some feature vector $\widetilde{\bx} \in \mathcal{X}_{T} = \mathcal{X}_{S}$, and $\btheta_{T}$ unknown.  Leveraging the RECaST transfer learning framework, the likelihood function of $(\delta, \gamma)$ can be expressed as
\begin{equation}\label{equation:joint_likelihood}
L(\delta, \gamma \mid y_{T,1}, \dots, y_{T,n_{T}}, \beta_{1}, \dots, \beta_{n_{T}}) = \prod_{i=1}^{n_{T}}\Big[\mathcal{N}\big\{ y_{T,i} \mid \beta_{i} \widetilde{ \bx }^{\top}\btheta_S, \sigma^{2} \big\} \cdot \text{Cauchy}( \beta_{i} \mid \delta, \gamma )\Big].
\end{equation}

We investigate the asymptotic coverage of prediction sets constructed from the RECaST posterior predictive distribution with plugin maximum likelihood estimators (MLEs) $\widehat{\delta}$ and $\widehat{\gamma}$ for $\delta$ and $\gamma$, respectively:
\begin{equation*}
% \label{equation: estimated posterior predictive}
\pi(\widetilde{y}_{T}, \widetilde{ \beta } \mid y_{1,} \dots, y_{n_{T}}) = \mathcal{N}(\widetilde{y}_{T} \mid \widetilde{\beta} \widetilde{ \bx }^{\top} \btheta_{S}, \sigma^{2}) \cdot \text{Cauchy}(\widetilde{ \beta } \mid \widehat{ \delta }, |\widehat{ \gamma }|).
\end{equation*}
This is the same as considering maximum a posteriori (MAP) estimators for $\delta$ and $\gamma$ with a flat prior $\pi(\delta, \gamma) \propto 1$, and the choice of prior is not so meaningful in the $n_{T} \to \infty$ setting.  Recall that in the RECaST framework the $\beta_{1},\dots,\beta_{n_{T}}$ that appear in the likelihood function in Equation \eqref{equation:joint_likelihood} are iid $\text{Cauchy}(\delta, \gamma)$ random effects.  Nonetheless, we demonstrate with Lemma \ref{lemma: convergence of delta and gamma} that the MLEs $\widehat{\delta}$ and $\widehat{\gamma}$ converge in probability to fixed points such that 
\[
\pi(\widetilde{Y}_{T}, \widetilde{ \beta } \mid y_{1,} \dots, y_{n_{T}}) \approx \mathcal{N}(\widetilde{Y}_{T} \mid \widetilde{\beta} \widetilde{ \bx }^{\top} \btheta_{S}, \sigma^{2}) \cdot \mathbf{1}\bigg\{\widetilde{\beta} = \frac{  \widetilde{ \bx }^{\top} \btheta_{T}  }{  \widetilde{ \bx }^{\top} \btheta_{S}  }\bigg\} = \mathcal{N}(\widetilde{Y}_{T} \mid \widetilde{\bx}^{\top} \btheta_{T}, \sigma^{2}),
\]
as desired.  This fact leads to our main theoretical result, Theorem \ref{theorem: convergence}, which establishes the asymptotic validity of $1 - \alpha$ level RECaST prediction sets of the form $[a_{n_{T}}^{\alpha}, b_{n_{T}}^{\alpha}]$, with
\begin{align*}
a_{n_{T}}^{\alpha} & := \Phi^{-1}(\alpha/2) \cdot \sigma +  \widetilde{ \beta } \cdot \widetilde{ \bx }^{\top} \btheta_{S} \ \ \text{and} \\
b_{n_{T}}^{\alpha} & := \Phi^{-1}(1-\alpha/2) \cdot \sigma +  \widetilde{ \beta } \cdot \widetilde{ \bx }^{\top} \btheta_{S},
\end{align*}
for any $\alpha \in (0,1)$ and $\widetilde{\beta} \sim \text{Cauchy}(\widehat{ \delta }, |\widehat{ \gamma }|)$.

\begin{lemma}\label{lemma: convergence of delta and gamma}
Assuming $Y_{T,1}, \dots, Y_{T,n_{T}} \overset{\text{iid}}{\sim} \mathcal{N}(\widetilde{\bx}^{\top} \btheta_{T}, \sigma^{2})$ and $\beta_{1},\dots,\beta_{n_{T}} \overset{\text{iid}}{\sim} \text{Cauchy}(\delta,\gamma)$, independently, the MLEs of $\delta$ and $\gamma$ for Equation \eqref{equation:joint_likelihood} satisfy 
\begin{align*}
\widehat{\delta} \longrightarrow \frac{  \widetilde{ \bx }^{\top} \btheta_{T}  }{  \widetilde{ \bx }^{\top} \btheta_{S}  } \quad \text{and} \quad \widehat{\gamma} \longrightarrow 0 
\end{align*}
in probability as $n_{T} \to \infty$.
\end{lemma}
\begin{theorem}\label{theorem: convergence}
Assume that $\widetilde{Y}_{T} \sim \mathcal{N}(\widetilde{\bx}^{\top} \btheta_{T}, \sigma^{2})$.  Then, for any $\alpha \in (0,1)$,
\[
P\Big(\widetilde{Y}_{T} \in [a_{n_{T}}^{\alpha}, b_{n_{T}}^{\alpha}] \Big) = \int_{a_{n_{T}}^{\alpha}}^{b_{n_{T}}^{\alpha}} \frac{1}{\sigma\sqrt{2\pi}}e^{-\frac{1}{2\sigma^{2}}(\widetilde{y}_{T} - \widetilde{\bx}^{\top} \btheta_{T})^{2}} d\widetilde{y}_{T} \longrightarrow 1 - \alpha
\]
in probability as $n_{T} \to \infty$.
\end{theorem}
In Section \ref{section: Simulation Study}, we provide empirical evidence that RECaST achieves near nominal coverage even in more practical, small $n_{T}$ settings, trained on target data that arise from both linear and non-linear models.  In the empirical investigations in Section \ref{section: Simulation Study}, we relax the assumptions of known $\sigma$ and the availability of repeated samples from a fixed feature vector $\widetilde{\bx}$.

\section{Binary Response Data}
\label{Section: Binary Response Data}

\subsection{Model and Estimation}
\label{subsection: binary model and estimation}

Assume that $Y_{S,1}, \dots, Y_{S,n_{S}}$ and $Y_{T,1}, \dots, Y_{T,n_{T}}$ are mutually independent, Bernoulli random variables generated according to source and target models, respectively, as expressed in Equation (\ref{equation: dge}). Also assume that an estimator $f(\widehat{\btheta}_{S},\bx)$ is available for any feature vector $\bx \in \X_{S} = \X_{T}$, where $\widehat{\btheta}_{S}$ is an estimator of $\btheta_{S}$ based on $Y_{S,1}, \dots, Y_{S,n_{S}}$.  In the binary response setting, a natural choice for the $h$ function in the RECaST model, defined by Equation \eqref{transfer_learning_model}, is the logistic model formulation,
\[
Y_{T,i} = \mathbf{1}\big[U_{T,i} < \text{expit}\big\{\beta_{i}\cdot f(\widehat{\btheta}_{S},\bx_{T,i})\big\}\big],
\]
with $U_{T,i} \sim \text{Uniform}(0,1)$ independently for $i \in \{1,\dots,n_{T}\}$ and $\beta_{i} \sim \text{Cauchy}(\delta, \gamma)$.

As in the continuous setting, the RECaST posterior distribution of the parameters can be constructed as
\begin{align*}
& \pi\big(\delta, \gamma \mid y_{T,1}, \dots, y_{T,n_{T}}, \widehat{\btheta}_{S}\big) \\
& = \int_{\R}\dots \int_{\R} \pi\big(\delta, \gamma, \beta_{1}, \dots, \beta_{n_{T}} \mid y_{T,1}, \dots, y_{T,n_{T}}, \widehat{\btheta}_{S}\big) \ d\beta_{1} \dots d\beta_{n_{T}} \\
%& \propto \pi(\delta, \gamma) \cdot \int_{\R}\dots \int_{\R} \prod_{i=1}^{n_{T}}\Big(\text{Bernoulli}\big[y_{T,i} \mid \text{expit}\big\{\beta_{i} f(\widehat{\btheta}_{S},\bx_{T,i})\big\}\big]\cdot \text{Cauchy}(\beta_{i} \mid \delta, \gamma)\Big) \ d\beta_{1} \dots d\beta_{n_{T}} \\
& \propto \pi(\delta, \gamma) \cdot \prod_{i=1}^{n_{T}}\int_{\R} \text{Bernoulli}\big[y_{T,i} \mid \text{expit}\big\{\beta_{i} f(\widehat{\btheta}_{S},\bx_{T,i})\big\}\big]\cdot \text{Cauchy}(\beta_{i} \mid \delta, \gamma) \ d\beta_{i},
\numberthis \label{equation: binary posterior}
\end{align*}
and the posterior predictive distribution of the label $\widetilde{Y}_{T}$ associated with some new target feature vector $\widetilde{\bx}_{T}$ can be derived as the marginal distribution of  
\begin{align*}
& \pi\big(\widetilde{y}_{T}, \widetilde{\beta}, \delta, \gamma \mid y_{T,1}, \dots, y_{T,n_{T}}, \widehat{\btheta}_{S} \big) \\
%& = \text{Bernoulli}\big[\widetilde{y}_{T} \mid \text{expit}\big\{\widetilde{\beta} f(\widehat{\btheta}_{S},\widetilde{\bx}_{T})\big\}\big] \cdot \pi\big(\widetilde{\beta}, \delta, \gamma \mid y_{T,1}, \dots, y_{T,n_{T}}, \widehat{\btheta}_{S} \big) \\
& = \text{Bernoulli}\big[\widetilde{y}_{T} \mid \text{expit}\big\{\widetilde{\beta} f(\widehat{\btheta}_{S},\widetilde{\bx}_{T})\big\}\big] \cdot \text{Cauchy}(\widetilde{ \beta } \mid \delta, \gamma) \cdot \pi\big(\delta, \gamma \mid  y_{T,1}, \dots, y_{T,n_{T}}, \widehat{\btheta}_{S}\big).
\numberthis \label{equation: binary posterior predictive}
\end{align*}
We specify a canonical prior on $(\delta, \gamma)$ as 
\[
\pi(\delta, \gamma ) = \mathcal{N}(\delta \mid 1, \sigma_{\delta}^{2}) \cdot \log \mathcal{N}(\gamma \mid a,b),
\]
with diffuse choices of the hyperparameters $\sigma_{\delta}, a, b$.

A $1 - \alpha$ level RECaST prediction credible set, denoted $\Gamma_{n_{T}}^{\alpha}$, for binary response values is constructed as 
\begin{equation}\label{equation:binary_prediction_set}
\Gamma_{n_{T}}^{\alpha} = 
\begin{cases}
\{0\}, & \text{if} \quad \widetilde{p} < 1 - \widetilde{p} \quad \text{and} \quad 1 - \alpha \le 1 - \widetilde{p} \\
\{1\}, & \text{if} \quad 1 - \widetilde{p} \leq \widetilde{p} \quad \text{and} \quad 1 - \alpha \le \widetilde{p} \\
\{0,1\}, & \text{else},
\end{cases}
\end{equation}
where $\widetilde{p} := \pi\big(\widetilde{y}_{T} = 1 \mid y_{T,1}, \dots, y_{T,n_{T}}, \widehat{\btheta}_{S} \big)$.

\subsection{Remarks on Implementation}
\label{subsection: binary remarks on implementation}

The RECaST transfer learning computations in the binary response setting follow analogously to those described in Section \ref{subsection: continuous remarks on implementation}.  For completeness, Algorithm \ref{algorithm: binary RECaST prediction} specifies the procedure we propose for drawing samples from the posterior predictive distribution described by Equation \eqref{equation: binary posterior predictive}.

\begin{algorithm}[H]
 \caption{\footnotesize RECaST posterior predictive sampling: binary response data}\label{algorithm: binary RECaST prediction}
 \textbf{Input:} $\widetilde{\bx}_{T}$, samples from $\pi\big(\delta, \gamma \mid y_{T,1}, \dots, y_{T,n_{T}}, \widehat{\btheta}_{S}\big)$, and sample sizes $n_{\text{post}}$, $n_{\beta}$, and $n_{Y}$ \\
 \textbf{Output:} A sample of values from $\pi\big(\widetilde{y} \mid y_{T,1}, \dots, y_{T,n_{T}}, \widehat{\btheta}_{S} \big)$
 \begin{algorithmic}
 \For{$i \gets 1$ to $n_{\text{post}}$}
     \State $\delta, \gamma \gets \text{random} \big\{ \pi\big( \delta , \gamma \mid y_{T,1}, \dots, y_{T,n_{T}}, \widehat{\btheta}_{S}\big)  \big\}$
     \For{$j \gets 1$ to $n_{\beta}$}
         \State $\widetilde{ \beta } \gets \text{random}\big\{ \text{Cauchy}(\delta, \gamma) \big\}$
         \For{$k \gets 1$ to $n_{Y}$}
             \State $\widetilde{Y}_{T} \gets \text{random}\Big(\text{Bernoulli}\big[\text{expit}\big\{\widetilde{\beta} f(\widehat{\btheta}_{S},\widetilde{\bx}_{T})\big\}\big]\Big)$
         \EndFor
     \EndFor
 \EndFor
 \end{algorithmic}
\end{algorithm}

\section{Simulation Study}\label{section: Simulation Study}

\subsection{Objectives and Setup}\label{Subsection: Objectives and Setup}

In this section, we examine the finite sample performance of RECaST through simulations on synthetic data.
We consider continuous and binary responses with source models corresponding to linear (RECaST LM) and logistic (RECaST GLM) regression, respectively, as well as a DNN (RECaST DNN) source model for both response types. 

We generate the synthetic data from linear and logistic regressions with source parameter vector $\btheta_{S}$ and target parameter vector $\btheta_{T}$, with $p = 50$ features (including an intercept).  The features are generated from the standard Gaussian distribution, $\bx_{S,i}, \bx_{T,j} \sim \mathcal{N}_{p-1}(\boldsymbol{0}, \bI_{p-1})$.  We fix the source data generating parameters $\btheta_{S}$.
The source data generating parameters are set to $\btheta_{S} = (-\ba, \bb)$ where $\ba, \bb \in \mathbb{R}^{25}$ have components independently sampled from $\text{Uniform}(0.75, 5)$ and then fixed for all simulations.
The {\em similarity} of source and target domains is controlled by choosing the value of $\sigma_{\text{TL}} > 0$ in constructing $\btheta_{T} = \btheta_{S} + \boldsymbol{\epsilon}$ with $\boldsymbol{\epsilon} \sim \mathcal{N}_{p}(\bzero, \sigma_{\text{TL}}^{2} \bI_p)$.  We consider values of $\sigma_{\text{TL}}^{2} \in \{ 0, 0.25, 1, 4\}$.  Setting $\sigma_{\text{TL}}^{2} = 0$ corresponds to $\btheta_{T} = \btheta_{S}$, i.e., no difference between the source and target distributions.  We fix the source sample size at $n_{S} = 1,000$, and vary the target sample size $n_{T}$ to examine performance when $p < n_{T}$ ($n_{T} = 100, 250$), $p$ is near $n_{T}$ ($n_{T} = 40, 60$), and $p > n_{T}$ ($n_{T} = 20$).  We simulate 300 source and target data sets for each of these 20 combinations of $\sigma^{2}_{\text{TL}}$ and $n_{T}$ values, and implement the estimation procedures described in Sections \ref{subsection: continuous remarks on implementation} and \ref{subsection: binary remarks on implementation}.  See Appendix \ref{appendix: MCMC Implementation Details} for additional details about the specifics of our implementations.

A baseline for comparison is constructed from training a DNN on the target data, only, without any transfer learning, and we compare RECaST to other state-of-the-art transfer learning approaches.  We build a DNN on the source data and fine-tune the last layer on the target data (Unfreeze DNN); this is often referred to as freezing the weights of the source DNN and unfreezing the last layer.  In the binary setting, we also compare RECaST to the regularized logistic regression (Wiens) approach of \cite{wiens2014}.  This approach uses the combined source and target EHR data to build a regularized model for disease prediction -- a similar application to the real data we consider in Section \ref{Section: eICU}, but with the disadvantage that Wiens requires access to the source data where RECaST does not.

Throughout this section, all DNN training proceeds by setting aside a portion of the training data to be used as a calibration data set.
The final DNN parameters are chosen from the epoch with the minimum calibration loss to improve generalizability to out-of-sample test sets.  Additional details/specifications for our DNN training procedures are provided in Appendix \ref{appendix: Neural Network Training Procedure}.

\subsection{Continuous Response Results}\label{sims_continuous_response}

Out-of-sample root mean squared errors (RMSEs) for all methods, averaged over 300 source and target data sets are presented in Table \ref{table: continuous metric results}.

As expected, the performance of DNN deteriorates as the target sample size decreases.  Interestingly, the RECaST RMSE values remain consistent for each value of $\sigma^{2}_{\text{TL}}$, regardless of sample size, suggesting that RECaST is appropriate even when the target sample size is so small as to preclude a target-only analysis.  Meanwhile, Unfreeze DNN exhibits an increase in RMSE for each value of $\sigma_{\text{TL}}^{2}$ as $n_{T}$ decreases.  As source and target become more dissimilar, both Unfreeze DNN and RECaST exhibit similar increases in RMSE.  In fact, with $n_{T} = 250$ and $\sigma_{\text{TL}}^{2}=4$, the target-only DNN outperforms both transfer learning methods.  This setting is the most prone to negative transfer: the target sample size is large enough to learn meaningful DNN parameters,  {\em and} the source and target data distributions differ greatly, making transfer difficult.  While in some settings with larger target sample sizes the Unfreeze DNN slightly outperforms RECaST, it has larger standard errors and fails to provide uncertainty quantification. 

Table \ref{table: continuous coverage results} and Figure \ref{figure: continuous coverage} summarize the performance of the prediction uncertainty quantification provided by our RECaST framework implementations. Table \ref{table: continuous coverage results} presents the empirical coverage for 95\% nominal level prediction sets for each simulation setting.  RECaST methods consistently provide empirical coverage at or slightly above nominal levels, supporting the use of RECaST for inference on out-of-sample target domain predictions.  Additionally, Figure \ref{figure: continuous coverage} plots empirical versus nominal coverage for the $\sigma_{\text{TL}}^{2} = 0.25,\ n_{T} = 100$ and $\sigma_{\text{TL}}^{2} = 4 ,\ n_{T} = 20$ settings at a grid of nominal levels.  The empirical coverages consistently achieve associated nominal levels or are slightly conservative.

\begin{table}[H]
\centering\small
\begin{tabular}{llllll}
$n_{T}$ & $\sigma_{\text{TL}}^{2}$ & DNN & RECaST LM & RECaST DNN & Unfreeze DNN\\ \hline
250 & 0 & 
	2.8(0.376) & \textbf{0.517(0.0265)} & 1.21(0.0901) & 0.583(0.0375)\\
& 0.25 & 
	2.94(0.37) & 3.6(0.428) & 3.77(0.4) & \textbf{2.76(0.421)}\\
& 1 & 
	\textbf{3.14(0.431)} & 7.14(0.861) & 7.23(0.838) & 5.46(0.891)\\
& 4 & 
\textbf{3.67(0.524)} & 14.3(1.73) & 14.3(1.71) & 11.0(1.78)\\
\hline 
100 & 0 & 
	8.89(1.62) & \textbf{0.516(0.0222)} & 1.22(0.0954) & 0.805(0.0945)\\
& 0.25 & 
	9.08(1.28) & 3.58(0.423) & 3.76(0.396) & \textbf{3.24(0.567)}\\
& 1 & 
	9.43(1.29) & \textbf{7.1(0.848)} & 7.19(0.826) & 6.26(1.13)\\
& 4 & 
	\textbf{10.7(1.52)} & 14.2(1.7) & 14.2(1.68) & 12.5(2.1)\\
\hline 
60 & 0 & 
	13.5(2.47) & \textbf{0.521(0.0246)} & 1.23(0.109) & 1.48(0.288)\\
& 0.25 & 
	13.4(1.61) & \textbf{3.58(0.431)} & 3.76(0.415) & 3.68(0.775)\\
& 1 & 
	14.1(1.77) & 7.1(0.874) & 7.2(0.863) & \textbf{6.83(1.31)}\\
& 4 & 
	16.4(2.57) & 14.2(1.75) & 14.2(1.75) & \textbf{13.3(1.97)}\\
\hline 
40 & 0 & 
	16.8(2.59) & \textbf{0.521(0.0244)} & 1.22(0.0884) & 1.78(0.61)\\
& 0.25 & 
	17.1(2.38) & \textbf{3.6(0.409)} & 3.78(0.404) & 4.1(1.13)\\
& 1 & 
	17.8(2.46) & \textbf{7.15(0.825)} & 7.25(0.827) & 7.57(2.15)\\
& 4 & 
	20.2(2.93) & \textbf{14.3(1.65)} & 14.3(1.67) & 14.7(3.04)\\
\hline 
20 & 0 & 
	20.5(1.79) & \textbf{0.535(0.031)} & 1.24(0.106) & 2.49(2.19)\\
& 0.25 & 
	21.0(1.76) & \textbf{3.69(0.438)} & 3.86(0.422) & 4.73(2.54)\\
& 1 & 
	21.7(2.02) & \textbf{7.33(0.891)} & 7.41(0.869) & 8.49(3.5)\\
& 4 & 
	24.4(2.74) & \textbf{14.6(1.79)} & 14.7(1.77) & 15.7(3.95)
\end{tabular}
\caption{\footnotesize Out-of-sample RMSE (standard error) averaged over 300 source and target data sets for each setting; the out-of-sample test sets generated with each of the 300 target data sets each contain 250 observations. Bold values denote the lowest RMSE for each row.}
\label{table: continuous metric results}
\end{table}
\begin{table}[H]
\centering\small
\begin{tabular}{llll}
$n_{T}$ & $\sigma_{\text{TL}}^{2}$ & RECaST LM & RECaST DNN \\ \hline
250 & 0 & 
	95.5(1.77)& 94.3(1.93)\\
& 0.25 & 
	94.8(1.89)& 94.7(1.89)\\
& 1 & 
	94.8(1.9)& 94.8(1.84)\\
& 4 & 
	94.7(1.96)& 94.7(1.92) \\\hline
100 & 0 & 
	95.9(1.75)& 94.3(2.07)\\
& 0.25 & 
	96.0(1.8)& 95.6(2.0)\\
& 1 & 
	95.9(1.75)& 95.8(1.82)\\
& 4 & 
	95.8(1.82)& 95.8(1.87) \\\hline
60 & 0 & 
	96.5(2.17)& 94.4(2.61)\\
& 0.25 & 
	96.4(1.88)& 96.2(1.83)\\
& 1 & 
	96.4(1.79)& 96.2(1.79)\\
& 4 & 
	96.3(1.8)& 96.2(1.77) \\\hline
40 & 0 & 
	97.1(2.14)& 94.4(3.3)\\
& 0.25 & 
	96.3(2.38)& 96.0(2.39)\\
& 1 & 
	96.1(2.61)& 96.0(2.52)\\
& 4 & 
	95.9(2.76)& 95.8(2.78) \\\hline
20 & 0 & 
	97.8(1.83)& 95.1(2.98)\\
& 0.25 & 
	96.8(2.56)& 96.7(2.78)\\
& 1 & 
	96.8(2.61)& 96.8(2.83)\\
& 4 & 
	96.7(2.66)& 96.6(2.88)
\end{tabular}
\caption{\footnotesize Empirical coverage (standard error) at the 95\% nominal level, averaged over 300 source and target data sets for each setting; the out-of-sample test sets generated with each of the 300 target data sets each contain 250 observations. All reported values are multiplied by 100.}
\label{table: continuous coverage results}
\end{table}
\begin{figure}[H]
    \centering
    \includegraphics[scale=.39]{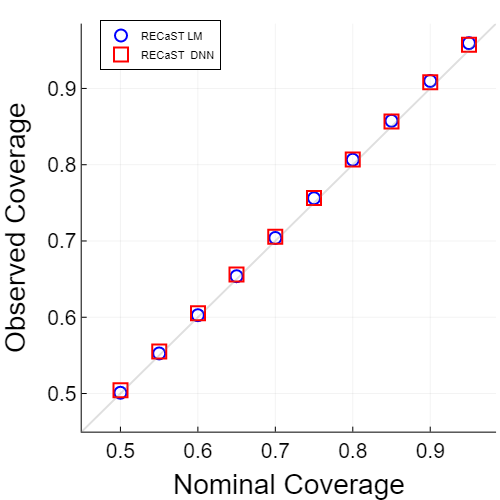}
    \includegraphics[scale=.39]{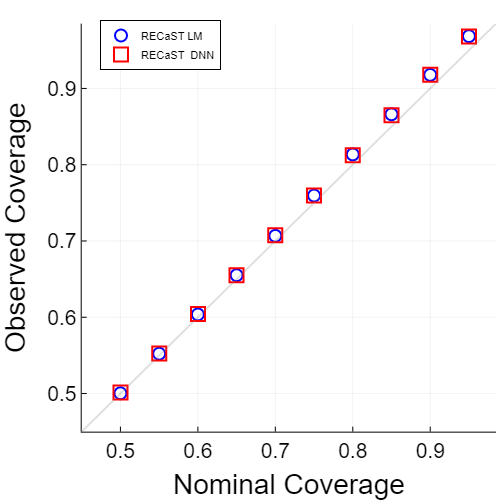}
    \caption{\footnotesize Reliability curves of the nominal coverage versus the empirical coverage, averaged over 300 source and target data sets for each setting; the out-of-sample test sets generated with each of the 300 target data sets each contain 250 observations. The left panel shows an easy setting: $n_{T} = 100$ and $\sigma_{\text{TL}}^{2} = 0.25$. The right panel shows a difficult setting: $n_{T} = 20$ and $\sigma_{\text{TL}}^{2} = 4$.}\label{figure: continuous coverage}
\end{figure}

\subsection{Binary Response Results}

Table \ref{table: discrete metric results} provides the area under the receiver operator characteristic curve (AUC) performance for all methods and simulation settings.  In all settings except one, RECaST DNN outperforms all other methods. We see similar patterns here as in the continuous setting. The RECaST models consistently report the highest AUC, with low standard errors across sample sizes. In contrast, the AUC of DNN and Unfreeze DNN drastically declines as $n_{T}$ decreases. As expected, the AUC of all transfer learning methods decreases as the difficulty of the problem increases with larger values of $\sigma^{2}_{\text{TL}}$. 
Table \ref{table: discrete coverage results} shows that RECaST procedures, again, provide near nominal coverages with low standard errors across sample sizes, while other methods fail.  Compared to the other approaches, RECaST provides substantial inferential advantages that are robust to small target sample sizes and large dissimilarity between source and target.  Recall from Equation \eqref{equation:binary_prediction_set} that prediction sets in the binary response setting are determined entirely by the Bernoulli probability of observing label 1.  Thus, we can construct prediction sets for the DNN, Unfreeze DNN, and Wiens methods, as well.  When a method fails to discriminate between the two labels at level $1 - \alpha$ (e.g., when the Bernoulli probability of success and failure are {\em both below} $1 - \alpha$), then the prediction set must include both labels to attain the $1 - \alpha$ level.  In such cases, as observed for the Wiens method in various settings in Table \ref{table: discrete coverage results}, the prediction set is guaranteed to achieve 100\% empirical coverage, but is unhelpful for prediction.

\begin{table}[H]
\centering\small
\begin{tabular}{lllllll}
$n_{T}$ & $\sigma_{\text{TL}}^{2}$ & DNN & RECaST GLM & RECaST DNN & Wiens & Unfreeze DNN\\ \hline
250 & 0 & 
	94.6(1.73) & 97.8(2.11) & \textbf{98.3(0.606)} & 80.4(3.54) & 97.4(1.21)\\
& 0.25 & 
	94.6(1.64) & 97.1(2.31) & \textbf{97.5(0.892)} & 79.9(3.87) & 96.9(1.16)\\
& 1 & 
	94.3(1.54) & 93.4(3.84) & \textbf{95.5(1.53)} & 79.1(3.87) & 95.1(1.66)\\
& 4 & 
	\textbf{94.5(1.66)} & 84.4(5.49) & 89.4(2.82) & 76.4(3.96) & 89.3(3.03)\\
\hline 
100 & 0 & 
	84.5(7.9) & 96.2(2.19) & \textbf{98.3(0.637)} & 80.7(4.17) & 96.1(2.21)\\
& 0.25 & 
	83.1(9.55) & 94.7(2.65) & \textbf{97.3(1.03)} & 80.0(4.14) & 94.9(1.91)\\
& 1 & 
	84.4(8.39) & 91.5(3.75) & \textbf{95.4(1.37)} & 78.5(4.43) & 93.2(2.4)\\
& 4 & 
	81.9(10.6) & 82.7(4.74) & \textbf{89.4(3.09)} & 74.3(4.33) & 86.7(4.8)\\
\hline 
60 & 0 & 
	72.2(12.7) & 95.8(1.87) & \textbf{98.0(1.0)} & 80.3(4.27) & 94.4(5.23)\\
& 0.25 & 
	73.8(10.7) & 94.1(2.46) & \textbf{97.1(1.44)} & 79.8(4.25) & 94.0(2.37)\\
& 1 & 
	74.6(10.3) & 90.3(3.53) & \textbf{95.2(1.65)} & 78.2(4.09) & 90.6(5.99)\\
& 4 & 
	72.3(10.6) & 82.7(3.97) & \textbf{88.8(3.33)} & 73.3(4.68) & 83.5(8.97)\\
\hline 
40 & 0 & 
	67.8(10.8) & 95.5(1.63) & \textbf{97.9(1.14)} & 80.1(3.82) & 93.5(4.47)\\
& 0.25 & 
	67.8(10.7) & 94.1(2.18) & \textbf{97.2(1.31)} & 79.8(4.0) & 92.0(6.69)\\
& 1 & 
	65.4(11.7) & 90.3(2.99) & \textbf{95.1(1.85)} & 78.1(3.88) & 88.9(7.93)\\
& 4 & 
	67.4(11.9) & 82.3(4.06) & \textbf{88.9(3.52)} & 73.5(4.22) & 80.0(11.5)\\
\hline 
20 & 0 & 
	60.2(8.68) & 95.7(1.72) & \textbf{97.3(1.35)} & 80.1(3.87) & 89.2(10.1)\\
& 0.25 & 
	60.8(9.11) & 93.5(2.07) & \textbf{96.5(1.87)} & 79.1(4.16) & 86.6(13.4)\\
& 1 & 
	59.7(9.46) & 89.8(2.73) & \textbf{94.3(2.47)} & 77.3(4.47) & 82.7(14.3)\\
& 4 & 
	62.3(8.2) & 82.2(3.46) & \textbf{88.2(3.03)} & 72.1(5.03) & 77.1(10.9)
\end{tabular}
\caption{\footnotesize Out-of-sample AUC (standard error) averaged over 300 source and target data sets for each setting; the out-of-sample test sets generated with each of the 300 target data sets each contain 250 observations. All reported values are multiplied by 100.  Bold values denote the highest AUC for each row.}
\label{table: discrete metric results}
\end{table}

\begin{table}[H]
\centering\small
\begin{tabular}{lllllll}
$n_{T}$ & $\sigma_{\text{TL}}^{2}$ & DNN & RECaST GLM & RECaST DNN & Wiens & Unfreeze DNN\\ \hline
250 & 0 & 
	84.1(9.51) & 95.4(0.78)& 95.5(0.0902)& 100.0(0.0) & 91.1(8.43)\\
& 0.25 & 
	88.6(7.51) & 95.2(0.819)& 95.5(0.133)& 100.0(0.0) & 93.4(6.67)\\
& 1 & 
	91.1(6.18) & 95.4(0.648)& 95.5(0.138)& 100.0(0.0) & 94.8(4.52)\\
& 4 & 
	93.0(6.4) & 95.3(0.399)& 95.2(0.392)& 98.7(1.53) & 94.5(4.75) \\\hline
100 & 0 & 
	79.6(11.8) & 95.6(1.09)& 95.7(0.251)& 100.0(0.0) & 90.4(8.66)\\
& 0.25 & 
	82.9(11.2) & 95.1(1.28)& 95.6(0.274)& 100.0(0.0) & 91.9(7.77)\\
& 1 & 
	88.2(8.21) & 94.9(1.23)& 95.6(0.346)& 99.5(4.61) & 93.6(5.9)\\
& 4 & 
	92.1(6.18) & 94.9(0.812)& 94.9(0.904)& 95.3(13.2) & 93.5(9.96) \\\hline
60 & 0 & 
	80.2(13.4) & 95.4(1.18)& 95.9(0.537)& 100.0(0.0) & 91.6(6.36)\\
& 0.25 & 
	76.8(16.7) & 95.1(1.27)& 95.8(0.671)& 100.0(0.0) & 91.1(8.13)\\
& 1 & 
	79.5(20.8) & 94.8(1.07)& 95.3(0.867)& 99.9(0.494) & 94.1(6.09)\\
& 4 & 
	84.3(14.7) & 95.2(0.585)& 94.5(1.01)& 96.8(4.68) & 92.6(8.67) \\\hline
40 & 0 & 
	68.4(22.5) & 95.3(1.61) & 95.8(0.855) & 100.0(0.0) & 88.5(11.1)\\
& 0.25 & 
	72.4(20.3) & 94.9(1.61) & 95.6(0.993) & 100.0(0.0) & 90.0(7.88)\\
& 1 & 
	75.6(19.4) & 94.0(1.52) & 95.0(1.15) & 99.8(0.526) & 89.1(8.69)\\
& 4 & 
	77.4(24.7) & 93.5(1.42) & 94.2(1.1) & 96.5(3.24) & 90.4(7.15) \\\hline
20 & 0 & 
	67.2(21.9) & 95.3(1.05) & 95.6(0.779) & 100.0(0.0) & 85.0(17.0)\\
& 0.25 & 
	75.3(15.8) & 94.9(1.08) & 95.4(0.978) & 100.0(0.0) & 85.8(13.5)\\
& 1 & 
	74.6(16.3) & 94.9(0.841) & 94.9(1.12) & 99.8(0.551) & 85.5(17.0)\\
& 4 & 
	71.6(13.0) & 94.5(0.466) & 94.2(0.993) & 98.3(1.5) & 79.6(17.6)
\end{tabular}
\caption{\footnotesize Empirical coverage (standard error) at the 95\% nominal level, averaged over 300 source and target data sets for each setting; the out-of-sample test sets generated with each of the 300 target data sets each contain 250 observations.
All reported values are multiplied by 100.}
\label{table: discrete coverage results}
\end{table}

\section{eICU Data}
\label{Section: eICU}

\subsection{Analysis Setup}

The eICU Collaborative Research Database \citep{pollard2018eicu} is a publicly available database of ICU encounters across multiple hospitals in the United States, making it well-suited for imitating transfer learning settings using real data.  In the spirit of the transfer learning application in \cite{wiens2014}, we focus on correctly diagnosing physiological shock for newly admitted ICU patients.  We define a binary response variable as the indicator of the event that a patient experienced shock upon ICU admission, using a combination of Internal Classification of Diseases 10 (ICD-10) codes: R57 Shock, not elsewhere classified; R58 Hemorrhage, not elsewhere classified; or R65.21 Severe sepsis with septic shock.  Features are limited to baseline variables measured at admission.  While the simulations of Section \ref{section: Simulation Study} explicitly link the source and target data through the data generation process, the similarity between source and targets defined in our eICU data application is unknown.

The data consist of measurements on 45,945 patients across 156 unique hospitals.  Only 700 of these patients were diagnosed with shock upon admission.  No individual hospital had enough positive cases to be reliably used as a source data set.  To curate a balanced data set, we take all 700 patients with shock and randomly sample an additional 700 patients with no shock.  Next, 80\% of the hospitals associated with our sampled 1,400 patients are randomly selected to define the `source hospital'.  The source data set consists of all ICU encounters at the `source hospital'.  Of the remaining 20\% of hospitals, half are randomly assigned to the `target training hospital', and the other half define a `target testing hospital'.  Notice that this procedure splits hospitals rather than patients; the source data set may not consist of 80\% of patients.  The target training and target testing data sets typically contain 80 to 130 patients each.

We repeat the described sampling procedure 300 times, to imitate 300 transfer learning scenarios from real data.  A logistic regression model and a DNN model are trained on each of the 300 source data sets, and all previously considered binary response transfer learning methods are implemented on the target data sets.  To boost the performance of the source DNN model, the architecture of the DNN is chosen from a set of candidate architectures by maximizing AUC, averaged over 100 of the source data sets; additional details are provided in Appendix \ref{appendix: Neural Network Training Procedure}.

\subsection{Results and Analysis}

In Figure \ref{eICU_sim_80_10_10}, we report the empirical coverage and AUC for all methods.
\begin{figure}[H]
\centering
\includegraphics[scale=.355]{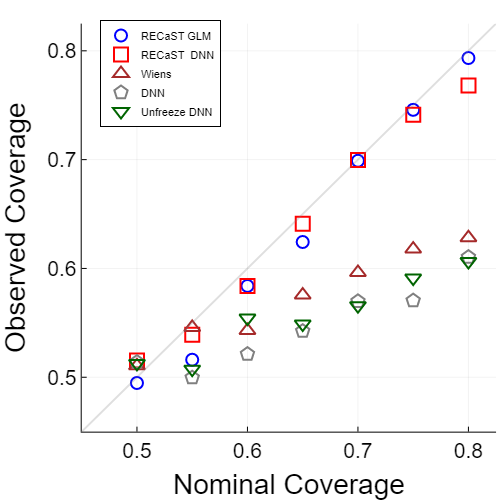}
\includegraphics[scale=.355]{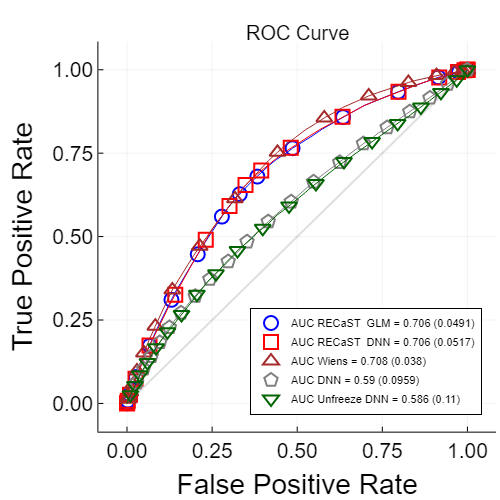}
\caption{\footnotesize The left panel displays the reliability curve of the nominal versus empirical out-of-sample coverage of prediction sets averaged over 300 target-testing data sets; the right panel reports the out-of-sample receiver operating characteristic (ROC) curve averaged pointwise over 300 target-testing data sets. The legend also reports the AUC (standard error) averaged over the same 300 target-testing data sets. Note that we cut the reliability curve at a nominal coverage of 0.8 because there are very few observations with higher coverage, undermining the reliability of coverage estimation at higher nominal levels.}
\label{eICU_sim_80_10_10}
\end{figure}
RECaST has similar predictive performance to Wiens but without requiring access to the source data, and it outperforms the DNN and Unfreeze DNN approaches.
Pairing RECaST with either the logistic regression or DNN source models produced near optimal average AUC, with respect to the average AUC values of 0.704 and 0.708, respectively, for the source logistic regression model and source DNN model.  Figure \ref{eICU_sim_80_10_10} also demonstrates that RECaST generally produces prediction sets that achieve their nominal level of coverage for target test response values, even for non-linear models with non-Gaussian data, whereas the other approaches do not.  This analysis demonstrates a general use case for RECaST as a clinical tool.

\section{Concluding Remarks}
\label{Section: conclusion}

We have demonstrated that the RECaST framework is adaptable to virtually any source model that makes predictions, and can accommodate both continuous and binary responses.  The source data themselves are not required, which is a significant advantage when legal or ethical barriers to access of source data sets exist, e.g., due to privacy concerns.  Unlike other transfer learning methods, RECaST always provides uncertainty quantification through prediction sets.  Our conclusions are supported by both theoretical justifications and performance in simulation studies on synthetic and real data.

The RECaST framework may be extended in several directions to accommodate the complexity of EHR data.
Broadening RECaST to handle differing feature spaces between source and target hospitals would allow for it to be applied in more general settings.  As EHR databases are updated, it would be useful to perform online transfer learning.  Patient clinical notes are also frequently available in EHR data and have been used by other transfer learning approaches \citep[e.g.,][]{si2020patient}. However, transfer learning approaches that combine quantitative and text features to create a unified patient representation are currently lacking.  Another promising research direction is to extend RECaST for trinomial label predictions \citep[see,][]{jacob2021,williams2021a}, count data, or time-to-event responses.
% \[
% h\big\{\beta_{i} f\big(\btheta_{S},\bx_{T,i}), U_{T,i}\big\} = \sum_{k=1}^{3} k\mathbf{1}\big[U_{T,i} \in \Delta_{k}\big\{\beta_{i} f(\btheta_{S},\bx_{T,i})\big\}\big],
% \]
% where $U_{T,i} \sim \text{Uniform}(\Delta)$ and $ \Delta_{1}$, $\Delta_{2}$, and $\Delta_{3}$, all functions of $\beta_{i} f(\btheta_{S},\bx_{T,i})$, are three triangular regions that form a partition of the triangle $\Delta$ \citep[see,][]{jacob2021,williams2021a}.  
%

% Acknowledgements and Disclosure of Funding should go at the end, before appendices and references
\section*{Acknowledgements}
Research reported in this publication was supported by the National Heart, Lung, And Blood Institute of the National Institutes of Health under Award Number R56HL155373. The content is solely the responsibility of the authors and does not necessarily represent the official views of the National Institutes of Health.

\appendix

% \section*{Appendices}
\addcontentsline{toc}{section}{Appendices}
\renewcommand{\thesubsection}{\Alph{subsection}}

\section{Proofs}\label{appendix: proof}

\begin{proof}[Proof of Lemma \ref{lemma: cauchy}]
It is well-established \citep[see, e.g.,][]{hinkley1969ratio} that if $V\sim \mathcal{N}(0, \sigma_{V}^{2})$ and $W \sim \mathcal{N}(0, \sigma_{W}^{2})$ with correlation coefficient $\rho$, then
\begin{align}\label{hinkley_cauchy}
\frac{V}{W} \sim \text{Cauchy}\Big( \frac{ \rho \sigma_{V} }{ \sigma_{W} } , \frac{ \sigma_{V} }{ \sigma_{W} } \sqrt{ 1 - \rho^{2} }\Big).
\end{align}
Accordingly, since
\[
\begin{bmatrix}
\ba^{\top} \\ \bb^{\top}
\end{bmatrix}
\bx
\sim 
\mathcal{N}\bigg( 
    \begin{bmatrix}
        0 \\
        0
    \end{bmatrix},
    \begin{bmatrix}
       \ba^{\top} \ba & \ba^{\top} \bb \\
        \bb^{\top} \ba & \bb^{\top} \bb
    \end{bmatrix}
\bigg),
\]
it follows that $\bx^{\top}\ba \sim \mathcal{N}(0, \ba^{\top}\ba)$, $\bx^{\top}\bb \sim \mathcal{N}( 0, \bb^{\top}\bb)$, and $\rho = (\ba^{\top} \bb ) / (\norm{\bb} \norm{\ba})$.  The result follows from Equation \eqref{hinkley_cauchy} by taking $V = \bx^{\top}\ba$ and $W = \bx^{\top}\bb$.
\end{proof}

Before proceeding directly to the proof of Lemma \ref{lemma: convergence of delta and gamma}, the following necessary supporting result is stated and proved.

\begin{lemma}\label{lemma: MLEs}
The MLEs of $\gamma$ and $\delta$ for Equation \eqref{equation:joint_likelihood}, respectively, are
\begin{align*}
\widehat{\gamma} & = \frac{  \sum_{i=1}^{n_{T}}(v_{i} - \overline{v})(y_{i} - \overline{y}_{T})  }{ \widetilde{ \bx }^{\top}\btheta_S \sum_{i=1}^{n_{T}}(v_{i} - \overline{v})^{2} } \ \ \text{and} \\
\widehat{\delta} & = \frac{\overline{y}_{T}}{\widetilde{ \bx }^{\top}\btheta_S} - \overline{v}\cdot \widehat{\gamma}, 
\end{align*}
where $v_{i} = (\beta_{i} - \delta)/\gamma$ for $i \in \{1,\dots,n_{T}\}$, $\overline{v} := \sum_{i=1}^{n_{T}}v_{i} / n_{T}$ and $\overline{y}_{T} := \sum_{i=1}^{n_{T}}y_{T,i} / n_{T}$.
\end{lemma}

\begin{proof}[Proof of Lemma \ref{lemma: MLEs}]
After the change of variables $v_{i} = (\beta_{i} - \delta)/\gamma$ for $i \in \{1,\dots,n_{T}\}$, the likelihood function in Equation \eqref{equation:joint_likelihood} takes the form
\begin{align*}
\prod_{i=1}^{n_{T}}\Big[ \mathcal{N}\big\{ y_{T,i} \mid (\gamma v_{i} + \delta) \widetilde{ \bx }^{\top}\btheta_S, \sigma^{2} \big\} \cdot \text{Cauchy}( v_{i} \mid 0, 1 ) \Big].
\end{align*}
Taking partial derivatives with respect to $\delta$ and $\gamma$ gives the first-order conditions
\begin{align*}
\sum_{i=1}^{n_{T}}\Big\{\frac{y_{T,i}}{\widetilde{ \bx }^{\top}\btheta_S} - \gamma v_{i} - \delta\Big\} & = 0 \\
\sum_{i=1}^{n_{T}}\Big\{\frac{y_{T,i}}{\widetilde{ \bx }^{\top}\btheta_S} - \gamma v_{i} - \delta\Big\}v_{i} & = 0.
\end{align*}
Solving this system yields the MLEs in Lemma \ref{lemma: MLEs}.
\end{proof}

\begin{proof}[Proof of Lemma \ref{lemma: convergence of delta and gamma}]
With the assumptions that $Y_{T,1}, \dots, Y_{T,n_{T}} \overset{\text{iid}}{\sim} \mathcal{N}(\widetilde{\bx}^{\top} \btheta_{T}, \sigma^{2})$ independent of $V_{1},\dots,V_{n_{T}} \overset{\text{iid}}{\sim} \text{Cauchy}(0,1)$, first, define the following notations: 
\[
\boldsymbol{Y} := 
\begin{pmatrix}
Y_{T,1} \\ 
\vdots \\
Y_{T,n_{T}}
\end{pmatrix}, \ \ \overline{\boldsymbol{Y}} := \overline{Y}_{T} \cdot \boldsymbol{1}_{n_{T}}, \ \ \overline{Y}_{T} := \frac{1}{n_{T}}\sum_{i=1}^{n_{T}}Y_{T,i},
\]
and
\[
\boldsymbol{V} := 
\begin{pmatrix}
V_{1} \\ 
\vdots \\
V_{n_{T}}
\end{pmatrix}, \ \ \overline{\boldsymbol{V}} := \overline{V} \cdot \boldsymbol{1}_{n_{T}}, \ \ \overline{V} := \frac{1}{n_{T}}\sum_{i=1}^{n_{T}}V_{i},
\]
where $\boldsymbol{1}_{n_{T}}$ is an $n_{T}$-dimensional column vector with every component having value 1.  

By the Cauchy-Schwarz inequality,
\[
|\widehat{ \gamma }| = \frac{ \Big| \sum_{i=1}^{n_{T}}(V_{i} - \overline{V})(Y_{i} - \overline{Y}_{T}) \Big| }{ \big| \widetilde{ \bx }^{\top}\btheta_S \big| \sum_{i=1}^{n_{T}}(V_{i} - \overline{V})^{2} } \leq \frac{ \norm{\bY - \overline{\bY} }_{2} \norm{ \boldsymbol{V} - \overline{\boldsymbol{ V }}}_{2} }{ \big| \widetilde{ \bx }^{\top}\btheta_S \big| \norm{\boldsymbol{V} - \overline{\boldsymbol{V}}}_{2}^{2}} = \frac{ \norm{\bY - \overline{\bY} }_{2} }{ \big| \widetilde{ \bx }^{\top}\btheta_S \big| \norm{\boldsymbol{V} - \overline{\boldsymbol{V}}}_{2}}, 
\]
where $\norm{\cdot}_{2}$ is the Euclidean norm.  We first need to establish the fact that square-root sums of independent, centered, and squared Cauchy random variables grow in value at the rate of at least $n_{T}^{\alpha+\frac{1}{2}}$ for any $\alpha \in (0,1/2)$.  Accordingly, for any $\varepsilon > 0$ and any $\alpha \in (0,1/2)$,
\begin{align*}
P\bigg( \norm{\boldsymbol{V} - \overline{\boldsymbol{V}}}_{2} < n_{T}^{\alpha+\frac{1}{2}} \varepsilon^{-1} \bigg) & = P\bigg( \norm{\boldsymbol{V} - \overline{\boldsymbol{V}}}_{2}^{2} < n_{T}^{2\alpha+1} \varepsilon^{-2} \bigg) \\
& = P\bigg( \sum_{i=1}^{n_{T}} V_{i}^{2} - n_{T} \overline{ V }^{2} < n_{T}^{2\alpha+1} \varepsilon^{-2} \bigg) \\
% & = P\bigg( \sum_{i=1}^{n_{T}} V_{i}^{2} - n_{T}^{1+\alpha} - n_{T} \overline{ V }^{2} < n_{T}^{2\alpha+1} \varepsilon^{-2} - n_{T}^{1+\alpha} \bigg) \\
& \le P\bigg( \sum_{i=1}^{n_{T}} V_{i}^{2} - n_{T}^{1+\alpha} < n_{T}^{2\alpha+1} \varepsilon^{-2} \bigg) + P\Big( - n_{T} \overline{ V }^{2} < -n_{T}^{1+\alpha} \Big) \\
& = P\bigg( \sum_{i=1}^{n_{T}} V_{i}^{2} < n_{T}^{2\alpha+1} \varepsilon^{-2} +  n_{T}^{1+\alpha} \bigg) + P\Big( |\overline{ V }| > n_{T}^{\frac{\alpha}{2}} \Big) \\
& \le P\bigg( \sum_{i=1}^{n_{T}} V_{i}^{2} < n_{T}^{2\alpha+1}\{\varepsilon^{-2} + 1\} \bigg) + 2 F_{V}\Big(-n_{T}^{\alpha / 2}\Big), \numberthis \label{equation:cauchybound}
\end{align*}
where $F_{V}(\cdot)$ is the $\text{Cauchy}(0,1)$ distribution function.  The first term vanishes for any $\alpha \in (0,1/2)$ as $n_{T} \to \infty$ by Lemma 2.1 in \cite{eicker1985}, and the second term vanishes as $n_{T} \to \infty$ by the definition of a distribution function.

Next, in order to show the convergence of both MLEs, we need that $n_{T}^{\alpha / 2}\widehat{\gamma} \to 0$ in probability as $n_{T} \to \infty$.  Our argument goes as follows.  For any $\varepsilon > 0$ and any $\alpha \in (0,1/2)$,
\begin{align*}
P\Big(|\widehat{\gamma}| > n_{T}^{-\alpha / 2} \varepsilon\Big) & \leq  P\Bigg( \frac{ \norm{\bY - \overline{\bY} }_{2} }{ \big| \widetilde{ \bx }^{\top}\btheta_S \big| \norm{\boldsymbol{V} - \overline{\boldsymbol{V}}}_{2}} > \frac{n_{T}^{ ( 1+\alpha ) / 2 }}{n_{T}^{ ( 1+\alpha ) / 2 }} \frac{\varepsilon}{n_{T}^{\alpha / 2}} \Bigg) \\
& \leq  P\Bigg( \frac{ \norm{\bY - \overline{ \bY }}_{2} }{ \big| \widetilde{ \bx }^{\top} \btheta_{S} \big|} > n_{T}^{ ( 1+\alpha ) / 2 } \Bigg) +  P\Bigg(\frac{ 1 }{ \norm{\boldsymbol{V} - \overline{\boldsymbol{V}}}_{2} } > \frac{ \varepsilon }{ n_{T}^{\alpha+1/2} } \Bigg)\\
%& = P\Bigg( \frac{ \norm{\bY - \overline{ \bY }}_{2}^{2} }{ \big| \widetilde{ \bx }^{\top} \btheta_{S} \big|^{2}} > n_{T}^{1+\alpha} \Bigg) + P\Big( \norm{\boldsymbol{V} - \overline{\boldsymbol{V}}}_{2} < n_{T}^{\alpha+1/2} \varepsilon^{-1} \Big)\\
& = P\Bigg( \frac{ \norm{\bY - \overline{ \bY }}_{2}^{2}}{\sigma^{2}} >  \frac{ \big| \widetilde{ \bx }^{\top} \btheta_{S} \big|^{2}}{\sigma^{2}}n_{T}^{1+\alpha} \Bigg) + P\bigg( \norm{\boldsymbol{V} - \overline{\boldsymbol{V}}}_{2} < n_{T}^{\alpha+\frac{1}{2}}\varepsilon^{-1} \bigg).
\end{align*}
Denoting $S := \norm{\bY - \overline{ \bY }}_{2}^{2} / \sigma^{2} \sim \chi_{n_{T}-1}^{2}$, and applying the Chernoff bound to the first quantity in the last expression gives, for any $t < 1/2$,
\[
P\Bigg( S >  \frac{ \big| \widetilde{ \bx }^{\top} \btheta_{S} \big|^{2}}{\sigma^{2}} n_{T}^{1+\alpha} \Bigg) 
% &= P\Bigg( \exp(tS) >  \exp \Bigg\{\frac{ t\big| \widetilde{ \bx }^{\top} \btheta_{S} \big|^{2}}{\sigma^{2}} n_{T}^{1+\alpha} \Bigg\} \Bigg) \\
\leq (1-2t)^{-(n_{T}-1)/2} \exp \bigg\{-\frac{ t\big| \widetilde{ \bx }^{\top} \btheta_{S} \big|^{2}}{\sigma^{2}} n_{T}^{1+\alpha}\bigg\}.
\]
Choosing $t = 1/4$ yields the bound
\[
P\Big(|\widehat{\gamma}| > n_{T}^{-\alpha / 2} \varepsilon\Big) \le e^{-n_{T}^{1+\alpha} \cdot \frac{1}{2}\big(\frac{1}{2\sigma^{2}}|\widetilde{ \bx }^{\top} \btheta_{S}|^{2} - n_{T}^{-\alpha} + n_{T}^{-1-\alpha}\big)} + P\bigg( \norm{\boldsymbol{V} - \overline{\boldsymbol{V}}}_{2} < n_{T}^{\alpha+\frac{1}{2}}\varepsilon^{-1} \bigg).
\]
%\begin{align*}
%P\Bigg( S >  \frac{ \big| \widetilde{ \bx }^{\top} \btheta_{S} \big|^{2}}{\sigma^{2}} n_{T}^{1+\alpha} \Bigg) %& \leq \Big(1-2 \frac{1}{4}\Big)^{-(n_{T}-1)/2} \exp \Bigg\{-\frac{ \frac{ 1 }{ 4 }\big| \widetilde{ \bx }^{\top} \btheta_{S} \big|^{2}}{\sigma^{2}} n_{T}^{1+\alpha}\Bigg\} \\
%& = \Big(\frac{1}{2}\Big)^{-(n_{T}-1)/2} \exp \Bigg\{-\frac{ \frac{ 1 }{ 4 }\big| \widetilde{ \bx }^{\top} \btheta_{S} \big|^{2}}{\sigma^{2}} n_{T}^{1+\alpha}\Bigg\} \\
%& = 2^{(n_{T}-1)/2} \exp\Bigg\{-\frac{ \frac{ 1 }{ 4 }\big| \widetilde{ \bx }^{\top} \btheta_{S} \big|^{2}}{\sigma^{2}} n_{T}^{1+\alpha}\Bigg\} \\
%& \leq \exp\{(n_{T}-1)/2\} \exp \Bigg\{-\frac{ \frac{ 1 }{ 4 }\big| \widetilde{ \bx }^{\top} \btheta_{S} \big|^{2}}{\sigma^{2}} n_{T}^{1+\alpha}\Bigg\} \\
%& = \exp\Bigg\{\frac{ n_{T} - 1 }{ 2 } -\frac{\big| \widetilde{ \bx }^{\top} \btheta_{S} \big|^{2}}{4\sigma^{2}} n_{T}^{1+\alpha}\Bigg\} \\
%& \le e^{-n_{T}^{1+\alpha} \cdot \frac{1}{2}\big(\frac{1}{2\sigma^{2}}|\widetilde{ \bx }^{\top} \btheta_{S}|^{2} - n_{T}^{-\alpha} + n_{T}^{-1-\alpha}\big)}
%\end{align*}
Thus, by Equation \eqref{equation:cauchybound}, it follows that $n_{T}^{\alpha / 2}\widehat{\gamma} \to 0$ in probability as $n_{T} \to \infty$.  This fact implies that $\widehat{\gamma} \to 0$ in probability as $n_{T} \to \infty$, and is needed to prove the asymptotic convergence of $\widehat{\delta}$, next. 

Since $Y_{T,1}, \dots, Y_{T,n_{T}} \overset{\text{iid}}{\sim} \mathcal{N}(\widetilde{\bx}^{\top} \btheta_{T}, \sigma^{2})$, it follows that $\overline{Y}_{T} = \widetilde{ \bx }^{\top}\btheta_{T} + \sigma n_{T}^{-\frac{1}{2}}U$, where $U \sim \mathcal{N}(0, 1)$.  That being so, for any $\varepsilon > 0$ and any $\alpha \in (0,1/2)$,
\begin{align*}
P\bigg( \Big| \widehat{ \delta } - \frac{ \widetilde{ \bx }^{\top} \btheta_{T} }{ \widetilde{ \bx }^{\top}\btheta_{S} } \Big| > \varepsilon \bigg) & = P\bigg\{ \Big|\frac{1}{\widetilde{ \bx }^{\top}\btheta_S} \big( \widetilde{ \bx }^{\top}\btheta_{T} + \sigma n_{T}^{-\frac{1}{2}} U \big) - \overline{V} \widehat{\gamma} - \frac{ \widetilde{ \bx }^{\top} \btheta_{T} }{ \widetilde{ \bx }^{\top}\btheta_{S} } \Big| > \varepsilon \bigg\} \\
& = P\bigg( \Big| \frac{ \sigma }{ \widetilde{ \bx }^{\top} \btheta_{S} } n_{T}^{-\frac{1}{2}}U - \overline{ V } \widehat{ \gamma } \Big| > \varepsilon\bigg) \\
& \leq P\bigg( \Big| \frac{ \sigma }{ \widetilde{ \bx }^{\top} \btheta_{S} } n_{T}^{-\frac{1}{2}}U \Big| > \frac{ \varepsilon }{ 2 }\bigg) + P\Big( | \overline{ V } \widehat{ \gamma } | > \frac{ \varepsilon }{ 2 }\Big) \\
%& = P\bigg( |U| > n_{T}^{\frac{1}{2}} \cdot \frac{ \varepsilon }{ 2 } \frac{ |\widetilde{ \bx }^{\top} \btheta_{S}| }{ \sigma } \bigg) + P\Big( | \overline{ V } \widehat{ \gamma } | > \frac{ \varepsilon }{ 2 }\Big) \\
& = 2\Phi\Big\{ - n_{T}^{\frac{1}{2}} \cdot \varepsilon |\widetilde{ \bx }^{\top} \btheta_{S}| / (2\sigma) \Big\} + P\Big( | \overline{ V } | > n_{T}^{\alpha / 2} / 2 \Big) + P\Big( | \widehat{ \gamma } | > n_{T}^{-\alpha / 2} \varepsilon \Big) \\
& = 2\Phi\Big\{ - n_{T}^{\frac{1}{2}} \cdot \varepsilon |\widetilde{ \bx }^{\top} \btheta_{S}| / (2\sigma) \Big\} + 2 F_{V}\Big(-n_{T}^{\alpha / 2}/2\Big) + P\Big(|\widehat{\gamma}| > n_{T}^{-\alpha / 2} \varepsilon\Big),
\end{align*}
where $\Phi(\cdot)$ is the standard Gaussian distribution function.  The first two terms in the last expression vanish by the definition of a distribution function, and the third term vanishes by the same because we previously established that $n_{T}^{\alpha / 2}\widehat{\gamma} \to 0$ in probability as $n_{T} \to \infty$.  Hence, $\widehat{ \delta } \to \widetilde{ \bx }^{\top} \btheta_{T} / (\widetilde{ \bx }^{\top} \btheta_{S}) $ in probability as $n_{T} \to \infty$.
\end{proof}
\begin{proof}[Proof of Theorem \ref{theorem: convergence}]
Our argument begins with direct evaluation of the probability that $\widetilde{Y}_{T} \sim \mathcal{N}(\widetilde{\bx}^{\top} \btheta_{T}, \sigma^{2})$ is contained in the interval $[a_{n_{T}}^{\alpha}, b_{n_{T}}^{\alpha}]$, and it finishes by applying the result of Lemma \ref{lemma: convergence of delta and gamma}.
\begin{align*}
P\Big(\widetilde{Y}_{T} \in [a_{n_{T}}^{\alpha}, b_{n_{T}}^{\alpha}] \Big) & = \int_{a_{n_{T}}^{\alpha}}^{b_{n_{T}}^{\alpha}} \frac{1}{\sigma\sqrt{2\pi}}e^{-\frac{1}{2\sigma^{2}}(\widetilde{y}_{T} - \widetilde{\bx}^{\top} \btheta_{T})^{2}} d\widetilde{y}_{T} \\
%& = \int_{\frac{a_{n_{T}}^{\alpha}-\widetilde{\bx}^{\top} \btheta_{T}}{\sigma}}^{\frac{b_{n_{T}}^{\alpha}-\widetilde{\bx}^{\top} \btheta_{T}}{\sigma}} \frac{1}{\sqrt{2\pi}}e^{-\frac{1}{2}z^{2}} dz \\
& = \Phi\Big( \frac{ b_{n_{T}}^{\alpha} - \widetilde{ \bx }^{\top} \btheta_{T} }{ \sigma }\Big) - \Phi\Big( \frac{ a_{n_{T}}^{\alpha} - \widetilde{ \bx }^{\top} \btheta_{T} }{ \sigma }\Big),
\end{align*}
where $\Phi(\cdot)$ is the standard Gaussian distribution function.  We will first demonstrate that $\Phi(W) \to 1 - \alpha/2$, with
\begin{align*}
W & := \frac{ b_{n_{T}}^{\alpha} - \widetilde{ \bx }^{\top} \btheta_{T} }{ \sigma } \\
& = \Phi^{-1}\Big(1 - \frac{ \alpha }{ 2 }\Big) + \frac{ 1 }{ \sigma } \Big(\widetilde{ \beta } \cdot \widetilde{ \bx }^{\top} \btheta_{S} - \widetilde{ \bx }^{\top} \btheta_{T}\Big) \\
& \sim \text{Cauchy}\bigg\{\Phi^{-1}\Big(1 - \frac{ \alpha }{ 2 }\Big) + \frac{ 1 }{ \sigma } \Big(\widehat{ \delta } \cdot \widetilde{ \bx }^{\top} \btheta_{S} - \widetilde{ \bx }^{\top} \btheta_{T}\Big), \ \Big|\frac{ \widehat{ \gamma }  }{ \sigma } \widetilde{ \bx }^{\top} \btheta_{S}\Big| \bigg\}
\end{align*}
since $\widetilde{ \beta } \sim \text{Cauchy}(\widehat{ \delta }, |\widehat{ \gamma }|)$.

For any $\epsilon > 0$,
\begin{align*}
P\Big( |\Phi(W) - (1 - \alpha/2)| < \epsilon \Big) & = P\Big( 1 - \alpha/2 -\epsilon < \Phi(W) < 1 - \alpha/2 + \epsilon \Big) \\
& = P\Big\{ \Phi^{-1}(1 - \alpha/2 -\epsilon) < W < \Phi^{-1}(1 - \alpha/2 + \epsilon) \Big\} \\
& = F_{W}\Big\{\Phi^{-1}(1 - \alpha/2 + \epsilon) \Big\} - F_{W}\Big\{\Phi^{-1}(1 - \alpha/2 - \epsilon) \Big\},
\end{align*}
where $F_{W}(\cdot)$ is the Cauchy distribution function associated with $W$.  Then,
\[
F_{W}\Big\{\Phi^{-1}(1 - \alpha/2 + \epsilon) \Big\} 
%& = \frac{1}{2} + \frac{1}{\pi} \arctan\Bigg[ \frac{ \Phi^{-1}(1 - \frac{ \alpha }{ 2 } + \epsilon) - \big\{  \Phi^{-1}(1 - \frac{ \alpha }{ 2 }) + (\widehat{ \delta } \cdot \widetilde{ \bx }^{\top} \btheta_{S} - \widetilde{ \bx }^{\top} \btheta_{T}) / \sigma \big\} }{ |\widehat{ \gamma }\cdot \widetilde{ \bx }^{\top} \btheta_{S} | / \sigma} \Bigg] \\
= \frac{1}{2} + \frac{1}{\pi} \arctan\Bigg\{ \frac{ c_{1} - (\widehat{ \delta } \cdot \widetilde{ \bx }^{\top} \btheta_{S} - \widetilde{ \bx }^{\top} \btheta_{T}) / \sigma }{ |\widehat{ \gamma }\cdot\widetilde{ \bx }^{\top} \btheta_{S} | / \sigma} \Bigg\}, \\
\]
with $c_{1} := \Phi^{-1}(1 - \alpha/2 + \epsilon) - \Phi^{-1}(1 - \alpha/2) > 0$, and similarly,
\[
F_{W}\Big\{\Phi^{-1}(1 - \alpha/2 - \epsilon) \Big\} = \frac{1}{2} + \frac{1}{\pi} \arctan\Bigg\{ \frac{ c_{2} - (\widehat{ \delta } \cdot \widetilde{ \bx }^{\top} \btheta_{S} - \widetilde{ \bx }^{\top} \btheta_{T}) / \sigma }{ |\widehat{ \gamma }\cdot\widetilde{ \bx }^{\top} \btheta_{S} | / \sigma} \Bigg\}, \\
\]
with $c_{2} := \Phi^{-1}(1 - \alpha/2 - \epsilon) - \Phi^{-1}(1 - \alpha/2) < 0$.  Accordingly, it follows by Lemma \ref{lemma: convergence of delta and gamma} that
\[
F_{W}\Big\{\Phi^{-1}(1 - \alpha/2 + \epsilon) \Big\} \longrightarrow 1 \quad \text{and} \quad F_{W}\Big\{\Phi^{-1}(1 - \alpha/2 - \epsilon) \Big\} \longrightarrow 0
\]
in probability as $n_{T} \to \infty$, and so 
\[
\Phi\Big( \frac{ b_{n_{T}}^{\alpha} - \widetilde{ \bx }^{\top} \btheta_{T} }{ \sigma }\Big) = \Phi(W) \longrightarrow 1 - \alpha/2
\]
in probability as $n_{T} \to \infty$.  A similar argument shows that 
\[
\Phi\Big( \frac{ a_{n_{T}}^{\alpha} - \widetilde{ \bx }^{\top} \btheta_{T} }{ \sigma }\Big) \longrightarrow \alpha/2,
\]
in probability as $n_{T} \to \infty$, concluding the proof.
\end{proof}

\section{Bounding Continuous Integral}
\label{appendix: Bounding Continuous Integral}

Recall the posterior distribution of the calibration parameters for the continuous response setting,
\begin{align*}
& \pi\big(\delta, \gamma , \sigma \mid y_{T,1}, \dots, y_{T,n_{T}}, \widehat{\btheta}_{S}\big) \\
& = \pi(\delta, \gamma, \sigma) \cdot \prod_{i=1}^{n_{T}} \int_{\R} \frac{ \text{Cauchy}( \beta_{i} \mid \delta, \gamma ) }{ \mid f(\widehat{\btheta}_S, \bx_{T,i}) \mid } \cdot \mathcal{N}\bigg\{ \beta_{i} \ \mid \ \frac{  y_{T,i} }{ f(\widehat{\btheta}_{S}, \bx_{T,i}) }, \frac{ \sigma^{2} }{ f(\widehat{\btheta}_S, \bx_{T,i})^{2}  } \bigg\} \ d\beta_{i}.
\end{align*}
Calculating this posterior requires the evaluation of $n_{T}$ integrals over $\mathbb{R}$.  For computational efficiency, we estimate the posterior by integrating over closed intervals.  The incurred numerical error can be tuned to be lower than computer precision.

Performing the substitution $u_{i} = \big\{\beta_{i} - y_{T,i}/f(\widehat{ \btheta }_{S}, \bx_{T,i}) \big\} / \big\{ \sigma / | f(\widehat{ \btheta }_{S}, \bx_{T,i}) | \big\}$ re-expresses the $i$th integral as
\begin{align*}
& \int_{\R} 
    \frac{\mathcal{N}(u_{i} \mid 0, 1) }{ \sigma } \cdot 
     \text{Cauchy}\Bigg[  u_{i} \ \mid 
        \ \frac{ | f(\widehat{ \btheta }_{S}, \bx_{T,i}) | }{ \sigma } \bigg\{ \delta - \frac{ y_{T,i} }{ f(\widehat{ \btheta }_{S}, \bx_{T,i}) } \bigg\}
        ,\  \frac{ | f(\widehat{ \btheta }_{S}, \bx_{T,i}) | \gamma }{ \sigma } \Bigg]
du_{i} \\
    & \leq \frac{ 1 }{ \sigma } \Bigg(\int_{s_{1}}^{s_{2}}
            \mathcal{N}(u_{i} \mid 0, 1) \cdot 
             \text{Cauchy}\Bigg[  u_{i}\  \mid \
                \frac{ | f(\widehat{ \btheta }_{S}, \bx_{T,i}) | }{ \sigma } \bigg\{ \delta - \frac{ y_{T,i} }{ f(\widehat{ \btheta }_{S}, \bx_{T,i}) } \bigg\}
                , \  \frac{ | f(\widehat{ \btheta }_{S}, \bx_{T,i}) | \gamma }{ \sigma } \Bigg]
                du_{i} \\
    & \hspace{.5in} + \phi(s_{1}) + \phi(s_{2}) \Bigg),
\end{align*}
for any $s_{1}$ and $s_{2}$ satisfying $s_{1} < 0 < s_{2}$, where $\phi(\cdot)$ is the standard Gaussian density function.  Then choose $s_{1}$ and $s_{2}$ so that $\phi(s_{1}) + \phi(s_{2})$ is as small as desired.  For example, we set $s_{1} = -39$ and $s_{2} = 39$, giving $\phi(s_{1})$ and $\phi(s_{2})$ numerically equal to zero in the base \texttt{Julia} software (for comparison, $\phi(38) = 1.097\times 10^{-314}$).

\section{MCMC Implementation Details}
\label{appendix: MCMC Implementation Details}
Sections \ref{subsection: continuous remarks on implementation} and \ref{subsection: binary remarks on implementation} detail the procedure for sampling from the posterior predictive distribution of a new observation.
RECaST first estimates the joint posterior density of the re-calibration parameters $(\delta, \gamma, \sigma)$ in the linear model and $(\delta,\gamma)$ in the logistic model.
We specify disperse priors $\delta \sim \mathcal{N}(1, 400)$, $\log(\gamma) \sim \mathcal{N}(0, 9)$, and in the continuous setting $\log(\sigma^{2})  \sim \mathcal{N}(0,9)$.
We run the Metropolis-Hastings estimation algorithm of the posterior distribution for 100,000 iterations with the initial 20,000 iterations used as a burn-in period to tune the proposal variance.
The parameters from the final 50,000 iterations are used as the posterior distribution.
Finally, $n_{\text{post}}=300$ equally spaced triplets/pairs of this distribution are taken as a posterior sample to be used in the posterior predictive estimation, which we denote by $\{\delta_i,\gamma_i,\sigma_i|\}_{i=1}^{300}$ and $\{\delta_i,\gamma_i\}_{i=1}^{300}$ in the linear and logistic models respectively.
For each triplet/pair, a sample of $n_{\widetilde{ \beta }}= 300$ $\beta$'s are taken from the Cauchy distribution, each used to generate $n_{Y} = 300$ samples from the posterior predictive distribution.
This gives $300 \times 300 \times 300 = 27,000,000$ posterior predictive observations for each out-of-sample test point, $(Y_{T, \text{test}}, \widetilde{ \bx }_{T})$.

\section{Neural Network Training Procedure}
\label{appendix: Neural Network Training Procedure}
The following procedure is used to train all neural networks considered: the source DNN, the DNN trained only on target data, and the Unfreeze DNN.

We initialize the weights using Xavier initialization \citep{glorot2010understanding}.
The network is trained for 2500 epochs using the ADAM optimizer and an MSE loss.
A portion of the training data is set aside as an out-of-sample calibration set during training.
At each epoch, the training and calibration loss are tracked.
The final parameterization used is taken from the epoch with the lowest calibration loss to avoid overfitting to the training data set.

The candidate architectures ranged from networks with 316 parameters to 11,641 parameters with varied number of layers, layer sizes, activation functions, and dropout proportions.
The architecture described below was chosen as it had the best test set AUC on the eICU data of all considered architectures.
We use a two layer neural network with layer sizes $\ell_{1} = (p , 25)$ and $\ell_{2} = (25, 1)$.
These layers are connected with a Rectified Linear Unit (ReLU) activation function.
In the binary response setting, the output of $\ell_{2}$ is converted to a probability through a softmax activation function.
For consistency, this architecture is also used for the simulated data analysis in Section \ref{section: Simulation Study}.

The source neural network for RECaST learns parameters in both layers using only source data.
The DNN network learns parameters in both layers using only the target data.
The Unfreeze DNN network learns parameters in both layers first using only the source data.
Then, the target data are processed through the same neural network, re-training parameters in the second layer and leaving the first layer unchanged from the values learned on the source data set.

\bibliographystyle{apalike}
\bibliography{references}

\begin{thebibliography}{}

\bibitem[Abba et~al., 2022]{abba2022}
Abba, M.~A., Williams, J.~P., and Reich, B.~J. (2022).
\newblock A penalized complexity prior for deep bayesian transfer learning with
  application to materials informatics.
\newblock {\em arXiv}, arXiv:2201.12401.

\bibitem[Ahishakiye et~al., 2021]{AHISHAKIYE2021118}
Ahishakiye, E., {Van Gijzen}, M.~B., Tumwiine, J., Wario, R., and Obungoloch,
  J. (2021).
\newblock A survey on deep learning in medical image reconstruction.
\newblock {\em Intelligent Medicine}, 1(3):118--127.

\bibitem[Akaoka et~al., 2022]{akaoka2022limit}
Akaoka, Y., Okamura, K., and Otobe, Y. (2022).
\newblock {Limit theorems for quasi-arithmetic means of random variables with
  applications to point estimations for the Cauchy distribution}.
\newblock {\em Brazilian Journal of Probability and Statistics}, 36(2):385 --
  407.

\bibitem[Baxter, 1998]{baxter1998}
Baxter, J. (1998).
\newblock {\em Theoretical Models of Learning to Learn}, pages 71--94.
\newblock Springer US, Boston, MA.

\bibitem[Bueno et~al., 2020]{bueno2019}
Bueno, A., Ben{\'\i}tez, C., De~Angelis, S., D{\'\i}az~Moreno, A., and
  Ib{\'a}{\~n}ez, J.~M. (2020).
\newblock Volcano-seismic transfer learning and uncertainty quantification with
  bayesian neural networks.
\newblock {\em IEEE Transactions on Geoscience and Remote Sensing},
  58(2):892--902.

\bibitem[Cai and Wei, 2021]{cai2021transfer}
Cai, T.~T. and Wei, H. (2021).
\newblock {Transfer learning for nonparametric classification: Minimax rate and
  adaptive classifier}.
\newblock {\em The Annals of Statistics}, 49(1):100 -- 128.

\bibitem[Chandra and Kapoor, 2020]{chandra2020}
Chandra, R. and Kapoor, A. (2020).
\newblock Bayesian neural multi-source transfer learning.
\newblock {\em Neurocomputing}, 378:54--64.

\bibitem[Dai et~al., 2007]{dai2007}
Dai, W., Yang, Q., Xue, G.-R., and Yu, Y. (2007).
\newblock Boosting for transfer learning.
\newblock In {\em Proceedings of the 24th International Conference on Machine
  Learning}, ICML '07, pages 193--200, New York, NY, USA. Association for
  Computing Machinery.

\bibitem[Deng et~al., 2009]{deng2009}
Deng, J., Dong, W., Socher, R., Li, L.-J., Li, K., and Fei-Fei, L. (2009).
\newblock Imagenet: A large-scale hierarchical image database.
\newblock In {\em 2009 IEEE Conference on Computer Vision and Pattern
  Recognition}, pages 248--255.

\bibitem[Desautels et~al., 2017]{desautels2017}
Desautels, T., Calvert, J., Hoffman, J., Mao, Q., Jay, M., Fletcher, G.,
  Barton, C., Chettipally, U., Kerem, Y., and Das, R. (2017).
\newblock Using transfer learning for improved mortality prediction in a
  data-scarce hospital setting.
\newblock {\em Biomedical Informatics Insights}, 9.

\bibitem[Dube et~al., 2020]{dube2020}
Dube, P., Bhattacharjee, B., Petit-Bois, E., and Hill, M. (2020).
\newblock {\em Improving Transferability of Deep Neural Networks}, pages
  51--64.
\newblock Springer International Publishing, Cham.

\bibitem[Eicker, 1985]{eicker1985}
Eicker, F. (1985).
\newblock Sums of independent squared cauchy variables grow quadratically:
  Applications.
\newblock {\em Sankhy{\=a}: The Indian Journal of Statistics, Series A
  (1961-2002)}, 47(1):133--140.

\bibitem[Fegyverneki, 2013]{fegyverneki2013simple}
Fegyverneki, S. (2013).
\newblock A simple robust estimation for parameters of cauchy distribution.
\newblock {\em Miskolc Math. Notes}, 14(3):887--892.

\bibitem[Freund and Schapire, 1999]{freund1999}
Freund, Y. and Schapire, R. (1999).
\newblock A short introduction to boosting.
\newblock {\em Japanese Society For Artificial Intelligence}, 14(771-780):1612.

\bibitem[Gao and Cui, 2021]{gao_cui2021}
Gao, Y. and Cui, Y. (2021).
\newblock Multi-ethnic survival analysis: Transfer learning with {C}ox neural
  networks.
\newblock In {\em Survival Prediction-Algorithms, Challenges and Applications},
  pages 252--257. PMLR.

\bibitem[Glorot and Bengio, 2010]{glorot2010understanding}
Glorot, X. and Bengio, Y. (2010).
\newblock Understanding the difficulty of training deep feedforward neural
  networks.
\newblock In Teh, Y.~W. and Titterington, M., editors, {\em Proceedings of the
  Thirteenth International Conference on Artificial Intelligence and
  Statistics}, volume~9 of {\em Proceedings of Machine Learning Research},
  pages 249--256, Chia Laguna Resort, Sardinia, Italy. PMLR.

\bibitem[Gong et~al., 2015]{gong2015}
Gong, J.~J., Sundt, T.~M., Rawn, J.~D., and Guttag, J.~V. (2015).
\newblock Instance weighting for patient-specific risk stratification models.
\newblock In {\em Proceedings of the 21th ACM SIGKDD International Conference
  on Knowledge Discovery and Data Mining}, KDD '15, pages 369--378, New York,
  NY, USA. Association for Computing Machinery.

\bibitem[Gu and Duan, 2022]{gu2022syntl}
Gu, T. and Duan, R. (2022).
\newblock Syntl: A synthetic-data-based transfer learning approach for
  multi-center risk prediction.
\newblock {\em medRxiv}.

\bibitem[Hector and Martin, 2022]{Hector-Martin}
Hector, E.~C. and Martin, R. (2022).
\newblock Turning the information-sharing dial: efficient inference from
  different data sources.
\newblock {\em arXiv}, arXiv:2207.08886.

\bibitem[Hinkley, 1969]{hinkley1969ratio}
Hinkley, D.~V. (1969).
\newblock {On the ratio of two correlated normal random variables}.
\newblock {\em Biometrika}, 56(3):635--639.

\bibitem[Jacob et~al., 2021]{jacob2021}
Jacob, P.~E., Gong, R., Edlefsen, P.~T., and Dempster, A.~P. (2021).
\newblock A gibbs sampler for a class of random convex polytopes.
\newblock {\em Journal of the American Statistical Association},
  116(535):1181--1192.
\newblock PMID: 35340357.

\bibitem[Johnson, 2013]{johnson2020quadgk}
Johnson, S.~G. (2013).
\newblock {QuadGK.jl}: {G}auss--{K}ronrod integration in {J}ulia.
\newblock \url{https://github.com/JuliaMath/QuadGK.jl}.

\bibitem[Kravchuk and Pollett, 2012]{kravchuk2012hodges}
Kravchuk, O.~Y. and Pollett, P.~K. (2012).
\newblock Hodges-lehmann scale estimator for cauchy distribution.
\newblock {\em Communications in Statistics - Theory and Methods},
  41(20):3621--3632.

\bibitem[Lee et~al., 2012]{lee2012}
Lee, G., Rubinfeld, I., and Syed, Z. (2012).
\newblock Adapting surgical models to individual hospitals using transfer
  learning.
\newblock In {\em 2012 IEEE 12th International Conference on Data Mining
  Workshops}, pages 57--63.

\bibitem[Li et~al., 2021]{li2021targeting}
Li, S., Cai, T., and Duan, R. (2021).
\newblock Targeting underrepresented populations in precision medicine: A
  federated transfer learning approach.
\newblock {\em arXiv}, arXiv:2108.12112.

\bibitem[Li et~al., 2022]{li2020transfer}
Li, S., Cai, T.~T., and Li, H. (2022).
\newblock Transfer learning for high-dimensional linear regression: Prediction,
  estimation and minimax optimality.
\newblock {\em Journal of the Royal Statistical Society. Series B, Statistical
  methodology}, 84(1):149---173.

\bibitem[Lu et~al., 2015]{lu2015}
Lu, J., Behbood, V., Hao, P., Zuo, H., Xue, S., and Zhang, G. (2015).
\newblock Transfer learning using computational intelligence: A survey.
\newblock {\em Knowledge-Based Systems}, 80:14--23.

\bibitem[Maurer et~al., 2015]{maurer2016}
Maurer, A., Pontil, M., and Romera-Paredes, B. (2015).
\newblock The benefit of multitask representation learning.
\newblock {\em Journal of Machine Learning Research}, 17.

\bibitem[Pan and Yang, 2010]{pan2009}
Pan, S.~J. and Yang, Q. (2010).
\newblock A survey on transfer learning.
\newblock {\em IEEE Transactions on Knowledge and Data Engineering},
  22(10):1345--1359.

\bibitem[Paul et~al., 2016]{paul2016}
Paul, R., Hawkins, S.~H., Balagurunathan, Y., Schabath, M.~B., Gillies, R.~J.,
  Hall, L.~O., and Goldgof, D.~B. (2016).
\newblock Deep feature transfer learning in combination with traditional
  features predicts survival among patients with lung adenocarcinoma.
\newblock {\em Tomography}, 2:388 -- 395.

\bibitem[Pollard et~al., 2018]{pollard2018eicu}
Pollard, T.~J., Johnson, A.~E., Raffa, J.~D., Celi, L.~A., Mark, R.~G., and
  Badawi, O. (2018).
\newblock The {eICU} collaborative research database, a freely available
  multi-center database for critical care research.
\newblock {\em Scientific Data}, 5(1):1--13.

\bibitem[Raghu et~al., 2019]{raghu2019}
Raghu, M., Zhang, C., Kleinberg, J., and Bengio, S. (2019).
\newblock Transfusion: Understanding transfer learning for medical imaging.
\newblock In Wallach, H., Larochelle, H., Beygelzimer, A., d`Alch\'{e} Buc, F.,
  Fox, E., and Garnett, R., editors, {\em Advances in Neural Information
  Processing Systems}, volume~32. Curran Associates, Inc.

\bibitem[Raina et~al., 2006]{raina2006}
Raina, R., Ng, A.~Y., and Koller, D. (2006).
\newblock Constructing informative priors using transfer learning.
\newblock In {\em Proceedings of the 23rd International Conference on Machine
  Learning}, ICML '06, pages 713--720, New York, NY, USA. Association for
  Computing Machinery.

\bibitem[Reeve et~al., 2021]{reeve2021adaptive}
Reeve, H. W.~J., Cannings, T.~I., and Samworth, R.~J. (2021).
\newblock {Adaptive transfer learning}.
\newblock {\em The Annals of Statistics}, 49(6):3618 -- 3649.

\bibitem[Schuster, 2012]{schuster2012parameter}
Schuster, S. (2012).
\newblock Parameter estimation for the cauchy distribution.
\newblock In {\em 2012 19th International Conference on Systems, Signals and
  Image Processing, IWSSIP 2012}, pages 350--353.

\bibitem[Shickel et~al., 2021]{shickel2021}
Shickel, B., Davoudi, A., Ozrazgat-Baslanti, T., Ruppert, M., Bihorac, A., and
  Rashidi, P. (2021).
\newblock Deep multi-modal transfer learning for augmented patient acuity
  assessment in the intelligent icu.
\newblock {\em Frontiers in Digital Health}, 3.

\bibitem[Si and Roberts, 2020]{si2020patient}
Si, Y. and Roberts, K. (2020).
\newblock Patient representation transfer learning from clinical notes based on
  hierarchical attention network.
\newblock {\em AMIA Summits on Translational Science Proceedings}, 2020:597.

\bibitem[Tian and Feng, 2022]{tian2022transfer}
Tian, Y. and Feng, Y. (2022).
\newblock Transfer learning under high-dimensional generalized linear models.
\newblock {\em Journal of the American Statistical Association}, 0(0):1--14.

\bibitem[Weiss et~al., 2016]{weiss2016}
Weiss, K., Khoshgoftaar, T.~M., and Wang, D. (2016).
\newblock A survey of transfer learning.
\newblock {\em Journal of Big data}, 3(1):9.

\bibitem[Wiens et~al., 2014]{wiens2014}
Wiens, J., Guttag, J., and Horvitz, E. (2014).
\newblock A study in transfer learning: leveraging data from multiple hospitals
  to enhance hospital-specific predictions.
\newblock {\em Journal of the American Medical Informatics Association},
  21(4):699--706.

\bibitem[Williams, 2021]{williams2021a}
Williams, J.~P. (2021).
\newblock Discussion of ``a gibbs sampler for a class of random convex
  polytopes''.
\newblock {\em Journal of the American Statistical Association},
  116(535):1198--1200.

\bibitem[Wohlert et~al., 2018]{wohlert2018}
Wohlert, J., Munk, A., Sengupta, S., and Laumann, F. (2018).
\newblock Bayesian transfer learning for deep networks.
\newblock {\em viXra}.

\bibitem[Wu et~al., 2017]{wu2017}
Wu, Q., Wu, H., Zhou, X., Tan, M., Xu, Y., Yan, Y., and Hao, T. (2017).
\newblock Online transfer learning with multiple homogeneous or heterogeneous
  sources.
\newblock {\em IEEE Transactions on Knowledge and Data Engineering},
  29(7):1494--1507.

\bibitem[Yang et~al., 2020]{yang2020}
Yang, H., Jiao, S., and Sun, P. (2020).
\newblock Bayesian-convolutional neural network model transfer learning for
  image detection of concrete water-binder ratio.
\newblock {\em IEEE Access}, 8:35350--35367.

\bibitem[Zhang et~al., 2017]{zhang2018}
Zhang, C., B{\"u}tepage, J., Kjellstr{\"o}m, H., and Mandt, S. (2017).
\newblock Advances in variational inference.
\newblock {\em IEEE Transactions on Pattern Analysis and Machine Intelligence},
  PP.

\bibitem[Zhang, 2010]{zhang2010highly}
Zhang, J. (2010).
\newblock A highly efficient l-estimator for the location parameter of the
  cauchy distribution.
\newblock {\em Computational statistics}, 25(1):97--105.

\bibitem[Zhao et~al., 2014]{zhao2014}
Zhao, P., Hoi, S.~C., Wang, J., and Li, B. (2014).
\newblock Online transfer learning.
\newblock {\em Artificial Intelligence}, 216:76--102.

\bibitem[Zhou et~al., 2020]{zhou2020}
Zhou, C., Zhang, J., Liu, J., Zhang, C., Shi, G., and Hu, J. (2020).
\newblock Bayesian transfer learning for object detection in optical remote
  sensing images.
\newblock {\em IEEE Transactions on Geoscience and Remote Sensing},
  58(11):7705--7719.

\end{thebibliography}

\end{document}